\documentclass[11pt]{amsart}

\usepackage[T1]{fontenc}
\usepackage{lmodern}
\usepackage{amsmath,amssymb,amsthm,mathtools}
\usepackage{microtype}

\allowdisplaybreaks

\newtheorem{theorem}{Theorem}[section]
\newtheorem{proposition}[theorem]{Proposition}
\newtheorem{lemma}[theorem]{Lemma}
\newtheorem{corollary}[theorem]{Corollary}
\theoremstyle{definition}
\newtheorem{definition}[theorem]{Definition}
\theoremstyle{remark}
\newtheorem{remark}[theorem]{Remark}

\newcommand{\Pm}{\mathbb P_m}
\newcommand{\Hm}{\mathbb H_m}
\newcommand{\Mm}{\mathbb M_m}
\newcommand{\Tr}{\operatorname{Tr}}
\newcommand{\altPhi}{\widehat{\Phi}}
\newcommand{\hatsigma}{\widehat{\sigma}}

\title[Alternative-mean trace divergences]
{Alternative-mean trace divergences: geometry, data processing,
	and barycenters}

\author{Trung Dung Vuong}

\address{
	High School for the Gifted, VNUHCM,
	153 Nguyen Chi Thanh, An Dong Ward,
	Ho Chi Minh City, Vietnam
}

\address{
	Vietnam National University Ho Chi Minh City,
	Linh Xuan Ward,
	Ho Chi Minh City, Vietnam
}

\email{vtdung@ptnk.edu.vn}

\author{Hiroki Shudo}

\address{Graduate school of science and engineering, Ritsumeikan University, Kusatsu, Shiga, 525-8577  Japan}

\email{ra0092px@ed.ritsumei.ac.jp}

\author{Hiroyuki Osaka }

\address{Department of Mathematical Sciences\\ Ritsumeikan University\\ Kusatsu, Shiga, 525-8577  Japan}

\email{osaka@se.ritsumei.ac.jp}


\begin{document}
	
\begin{abstract}
	Let $f:(0,\infty)\to(0,\infty)$ be a nontrivial normalized operator
	monotone function and set $s=f'(1)$. We introduce the alternative-mean
	trace functional
\[
\altPhi_f(A,B)
:=\Tr(A\nabla_s B)
-\Tr\!\left(
f(A^{-1}\sharp B)\,A\,f(A^{-1}\sharp B)
\right).
\]
	on the positive definite cone. We prove that $\altPhi_f$ is a quantum divergence in the sense of
	Bhatia--Gaubert--Jain whose diagonal Hessian induces a positive multiple
	of the Bures--Wasserstein Riemannian metric. We also establish the sharp
	comparison
	\[
	s(1-s)d_{\rm BW}(A,B)^2
	\le \altPhi_f(A,B)
	\le (1-s+s^2)d_{\rm BW}(A,B)^2.
	\] We further establish several equivalent conditions characterizing those
	$f$ for which $\altPhi_f$ satisfies data processing under all positive
	trace-preserving maps; in particular, diagonal pinching on faithful qubit
	states is a complete test. For the associated barycenter problem, we prove
	existence of closed-cone minimizers, derive a first-order equation for
	interior critical points, and give criteria ensuring uniqueness and positive
	definiteness. For the power family $f(x)=x^t$, we determine exactly when
	every finite data set in the positive definite cone admits a unique
	closed-cone minimizer, which is then necessarily positive definite.
\end{abstract}

	\maketitle
	
	\section{Introduction}
	
	For each integer $k\ge1$, let $\mathbb M_k:=\mathbb C^{k\times k}$, let
	$\mathbb H_k\subset\mathbb M_k$ be the real vector space of Hermitian
	matrices, and let $\mathbb P_k\subset\mathbb H_k$ be the open cone of
	positive definite matrices.  For $Y\in\Mm$, we denote its Hilbert--Schmidt norm by
	\[
	\|Y\|_{\rm HS}:=\bigl(\Tr(Y^*Y)\bigr)^{1/2}.
	\] For $A,B\in\Pm$, their geometric mean is
	\[
	A\sharp B
	:=A^{1/2}\bigl(A^{-1/2}BA^{-1/2}\bigr)^{1/2}A^{1/2}.
	\]
	A normalized operator monotone function
	$f:(0,\infty)\to(0,\infty)$ then determines the alternative mean
	\[
	A\hatsigma_fB
	:=f(A^{-1}\sharp B)A f(A^{-1}\sharp B),
	\qquad A,B\in\Pm.
	\]
A general theory of this construction was developed in
	\cite{DumitruFrancoKimCzerwinska2026}.  It includes, for
	$0\le t\le1$, the weighted spectral geometric mean
	\[
	A\natural_tB=(A^{-1}\sharp B)^tA(A^{-1}\sharp B)^t,
	\]
	corresponding to $f(x)=x^t$, and, for $0\le s\le1$, the two-variable
	Bures--Wasserstein mean
	\[
	A\diamond_sB=((1-s)I+s(A^{-1}\sharp B))
	A((1-s)I+s(A^{-1}\sharp B)),
	\]
	corresponding to $f(x)=(1-s)+sx$.
The former originates in the spectral geometric mean of Fiedler and Pt\'ak
\cite{FiedlerPtak1997}; see also
\cite{KimLee2015,GanTam2024,GanKimSpectral2024}.  The latter is the
two-variable least-squares mean for the Bures--Wasserstein distance
\cite{BhatiaJainLim2019,HwangKim2022}.  Recent relations and common
extensions of these two means have been studied in
\cite{GanKimWasserstein2024,GanHuang2024,GanKimMer2026,
	DumitruFrancoKim2025}.
	
	In this paper, \emph{nontrivial} always means that
	$f\not\equiv1$ and $f\not\equiv x$.  If $s=f'(1)$, then a nontrivial
	normalized operator monotone function has $s\in(0,1)$. We associate with $A\hatsigma_fB$ the functional
	\[
	\altPhi_f(A,B)
	:=\Tr(A\nabla_sB)-\Tr(A\hatsigma_fB),
	\qquad s=f'(1).
	\]
	Our first result establishes a sharp global two-sided comparison between
	$\altPhi_f$ and the squared Bures--Wasserstein distance. Since
	$\altPhi_f$ is smooth, this comparison implies that it is a quantum
	divergence in the sense of Bhatia--Gaubert--Jain. We then compute its diagonal Hessian and identify the induced local metric,
	determine exactly when $\altPhi_f$ satisfies the data-processing
	inequality under all positive trace-preserving maps, and develop the
	associated barycenter theory.
	   We use the term \emph{quantum divergence} in the sense of
Bhatia--Gaubert--Jain \cite{BGJ2019}: it refers to a smooth nonnegative
functional that vanishes only on the diagonal, whose first derivative in the
second variable vanishes there, and whose diagonal Hessian is positive
semidefinite.   
	
	\subsection*{Relation to recent work}
	Kubo--Ando trace divergences and generalized Hellinger divergences benefit from
	perspective structures that yield convexity and data processing
	\cite{BGJ2019,PitrikVirosztek2020,OsakaShudo2025}.  Here the nonlinear variable
	$A^{-1}\sharp B$ instead occurs inside an alternative mean, so those
	mechanisms do not apply directly.  The underlying alternative means were
	developed in \cite{DumitruFrancoKimCzerwinska2026}; related algebraic and
	linearization questions appear in
	\cite{GanKimMer2026,DumitruFrancoMerino2026,KimMer2026}. A related Bures--Wasserstein-type construction is the
	$\alpha$-$z$-divergence of Dinh--Le--Vo--Vuong
	\cite{DinhLeVoVuong2021}, which is built from the
	$\alpha$-$z$ R\'enyi trace functional rather than from an alternative
	mean.  In its admissible parameter range, it is a quantum divergence
	satisfying data processing, and its associated least-squares problem was
	also solved there.  At $\alpha=z=1/2$, it reduces to
	$\tfrac12d_{\rm BW}^2$, the same divergence obtained here from
	$f(x)=\sqrt{x}$.
	
For the power family $f_t(x)=x^t$, Gan--Jeong--Kim
	\cite{GanJeongKim2025} studied the associated divergence and barycenter
	for $0<t\le1/2$ and stated the CPTP data-processing inequality throughout
	this range.  Subsequent work showed that this inequality fails away from
	the midpoint: Le--Vu--Ho--Dinh proved that DPI fails under qubit pinching
	for every $t\in(0,1)\setminus\{1/2\}$
	\cite[Theorem~5.4]{LeVuHoDinh2026}.  Hiai's connection between CPTP
	monotonicity and joint concavity, together with his necessary parameter
	bound for the corresponding spectral trace functional, also forces
	$t=1/2$ upon setting $(\alpha,p)=(t,1)$
	\cite[Theorem~5.3(1) and Proposition~5.19(1)]{Hiai2026}.
	
	The two-parameter $(k,t)$-spectral divergence of Vuong--Dinh unifies this
	power family at $k=1/2$ with the divergence of Jeong--Kim--Tam
	\cite{JKT2025} at $t=1/2$ \cite[Remark~3.7]{VuongDinh2026}.  For
	$k\in(0,1)$ and $0<t\le1/2$, its barycenter exists uniquely in the
	positive definite cone \cite[Theorem~3.20]{VuongDinh2026}.  Its
	specialization at $k=1/2$ also provides a faithful qubit counterexample
	to DPI under diagonal pinching at $t=0.3$
	\cite[Proposition~4.3(b)]{VuongDinh2026}.

Our results extend this picture in two directions.  For data processing,
we treat every normalized operator monotone function $f$, classify the
entire exact-DPI class even under all positive trace-preserving maps, and
show that diagonal pinching on faithful qubit states is a complete test.
For barycenters, we show that the power-family range $0<t\le1/2$ is exact:
for every $1/2<t<1$, two commuting positive definite $2\times2$ data
matrices with equal weights already have nonunique closed-cone minimizers.	
	\subsection*{Main results}
	The paper has three connected parts.  The first places $\altPhi_f$ in
	Bures--Wasserstein geometry.  With $R=A^{-1}\sharp B$, both quantities are
	traces against $A$ of scalar functions of the same matrix:
	\[
	\altPhi_f(A,B)
	=\Tr A\bigl((1-s)I+sR^2-f(R)^2\bigr),
	\qquad
	d_{\rm BW}(A,B)^2=\Tr A(R-I)^2.
	\]
	The arithmetic and harmonic extremals give the optimal fixed-weight
	comparison
	\[
	s(1-s)d_{\rm BW}^2\le\altPhi_f
	\le(1-s+s^2)d_{\rm BW}^2.
	\]
	We also determine the best constants for each individual $f$.  At the
	diagonal, the Hessian is
	$2(s(1-s)-f''(1))$ times the Bures--Wasserstein metric.  Thus, each nontrivial normalized operator monotone function $f$ gives
	rise to a quantum divergence $\altPhi_f$.
	
The second and central part concerns data processing.
Theorem~\ref{thm:quadratic-square-classification} proves that, for
$\altPhi_f$, exact DPI under every completely positive trace-preserving
map is equivalent to exact DPI under every positive trace-preserving map.
Thus enlarging the class of admissible maps does not change the resulting
classification of $f$.
Either formulation can be tested by diagonal pinching on faithful qubit
states and holds exactly for
\[
f(x)=\sqrt{\bigl((1-u)+ux\bigr)\bigl((1-v)+vx\bigr)},
\qquad u,v\in[0,1].
\]
Equivalently, $f^2$ is a polynomial of degree at most two and
$\altPhi_f$ is a constant multiple of $d_{\rm BW}^2$.  The same class
is characterized by joint convexity of $\altPhi_f$, or equivalently by
joint concavity of the associated trace fidelity.  For every $f$ outside
this class, data processing fails under diagonal pinching for a pair of
real-symmetric faithful qubit states.
	
The third part develops the associated barycenter problem.  Closed-cone
minimizers always exist, and interior critical points satisfy the
corresponding first-order barycenter equation.  Scalar data have a unique barycenter without an additional
	hypothesis.  In dimension $m$, a nonnegative, nonconstant, order-$m$
	matrix-concave residual $f(x)^2-cx^2$ guarantees a unique positive definite
	minimizer; an operator-monotone residual gives a dimension-free criterion.
For $f(x)=x^t$, every data set has a unique closed-cone minimizer that is
positive definite exactly when $0<t\le1/2$.  Above the midpoint, it fails already for two
	commuting positive definite $2\times2$ matrices.
	
	Section~\ref{sec:sharp-bw} proves the sharp comparison, and
Section~\ref{sec:hessian} identifies the second-order
Bures--Wasserstein geometry.
	Section~\ref{sec:dpi} proves the data-processing rigidity theorem, and Section~\ref{sec:barycenters} is devoted to the associated
	barycenter problem.
	
	\section{Definition and sharp Bures--Wasserstein bounds}
	\label{sec:sharp-bw}
	
	For $A,B\in\Pm$ and $t\in[0,1]$, write
	\[
	A\nabla_tB=(1-t)A+tB,
	\qquad
	A\sharp B
	=A^{1/2}(A^{-1/2}BA^{-1/2})^{1/2}A^{1/2}.
	\]
If
\[
R:=A^{-1}\sharp B
=A^{-1/2}(A^{1/2}BA^{1/2})^{1/2}A^{-1/2},
\]
then $R$ is the unique positive definite solution of the Riccati equation
\begin{equation}\label{eq:riccati}
	RAR=B.
\end{equation}
	
	For a normalized operator monotone function $f$, the Kubo--Ando
	harmonic--arithmetic bounds state that, with $s=f'(1)$,
	\begin{equation}\label{eq:harm-arith}
		h_s(x)\le f(x)\le a_s(x),
		\qquad
		h_s(x):=\bigl((1-s)+sx^{-1}\bigr)^{-1},
		\quad
		a_s(x):=(1-s)+sx.
	\end{equation}
	Here $0\le s\le1$; if $f$ is nontrivial, then $0<s<1$.  We refer to
	\cite{KuboAndo1980} for these standard facts.
	
	\begin{definition}\label{def:divergence}
		Let $f:(0,\infty)\to(0,\infty)$ be a nontrivial normalized operator
		monotone function and set $s=f'(1)$.  Define
		\begin{equation}\label{eq:divergence}
			\altPhi_f(A,B)
			:=(1-s)\Tr A+s\Tr B
			-\Tr\!\left(f(A^{-1}\sharp B)A f(A^{-1}\sharp B)\right).
		\end{equation}
		For the endpoint functions $f\equiv1$ and $f(x)=x$, we use the same
		formula whenever they arise; in both cases $\altPhi_f\equiv0$.
	\end{definition}
	
	We use the squared Bures--Wasserstein distance
	\[
	d_{\rm BW}(A,B)^2
	:=\Tr A+\Tr B
	-2\Tr(A^{1/2}BA^{1/2})^{1/2};
	\]
	see \cite{BhatiaJainLim2019}.
	
	\begin{proposition}
		\label{prop:master-representation}
		Let $f$ be a normalized operator monotone function, set $s=f'(1)$, and
		let $A,B\in\Pm$.  Put $R=A^{-1}\sharp B$.  Then
		\begin{align}
			d_{\rm BW}(A,B)^2
			&=\Tr\!\left(A(R-I)^2\right), \label{eq:bw-R}\\
			\altPhi_f(A,B)
			&=\Tr\!\left(
			A\bigl((1-s)I+sR^2-f(R)^2\bigr)
			\right). \label{eq:phi-R}
		\end{align}
	\end{proposition}
	
	\begin{proof}
		Equation \eqref{eq:riccati} gives
		$\Tr B=\Tr(AR^2)$.  Moreover,
		\[
		(A^{1/2}BA^{1/2})^{1/2}=A^{1/2}RA^{1/2},
		\]
		and hence its trace is $\Tr(AR)$, proving \eqref{eq:bw-R}.  Finally,
		cyclicity of the trace gives
		$\Tr(f(R)Af(R))=\Tr(Af(R)^2)$, which yields
		\eqref{eq:phi-R}.
	\end{proof}
	
The next theorem gives the optimal comparison constants uniformly over all
normalized operator monotone functions with the prescribed weight
$s=f'(1)$.

\begin{theorem}
		\label{thm:sharp-comparison}
		Let $f$ be a nontrivial normalized operator monotone function and set
		$s=f'(1)\in(0,1)$.  For every $A,B\in\Pm$,
		\begin{equation}\label{eq:endpoint-comparison}
			\altPhi_{a_s}(A,B)
			\le\altPhi_f(A,B)
			\le\altPhi_{h_s}(A,B),
		\end{equation}
		and
		\begin{equation}\label{eq:sharp-sandwich}
			s(1-s)d_{\rm BW}(A,B)^2
			\le\altPhi_f(A,B)
			\le(1-s+s^2)d_{\rm BW}(A,B)^2.
		\end{equation}
		Both constants in \eqref{eq:sharp-sandwich} are optimal uniformly over all
		normalized operator monotone functions with derivative $s$.
		Moreover, for $A\ne B$, equality in the lower bound of
		\eqref{eq:sharp-sandwich} holds if and only if $f=a_s$.  Equality in the
		upper bound for a given pair $A\ne B$ holds if and only if
		\[
		f=h_s,\qquad s\ne\frac12,
		\qquad
		\sigma(A^{-1}\sharp B)
		\subseteq\left\{1,\frac{s^2}{(1-s)^2}\right\},
		\]
		with $s^2/(1-s)^2$ occurring in the spectrum.
	\end{theorem}
	
	\begin{proof}
		The three matrices $h_s(R),f(R),a_s(R)$ commute.  Thus
		\eqref{eq:harm-arith} implies
		\[
		h_s(R)^2\preceq f(R)^2\preceq a_s(R)^2,
		\]
		and \eqref{eq:phi-R} gives \eqref{eq:endpoint-comparison}.  Since
		\[
		(1-s)I+sR^2-a_s(R)^2=s(1-s)(R-I)^2,
		\]
		we obtain
		\begin{equation}\label{eq:arithmetic-endpoint}
			\altPhi_{a_s}(A,B)=s(1-s)d_{\rm BW}(A,B)^2.
		\end{equation}
		
		It remains to bound the harmonic endpoint.  For $x>0$,
		\[
		h_s(x)=\frac{x}{(1-s)x+s},
		\]
		and direct simplification gives
		\begin{align}
			&(1-s+s^2)(x-1)^2
			-\bigl((1-s)+sx^2-h_s(x)^2\bigr) \notag\\
			&\qquad
			=(x-1)^2
			\frac{\bigl((1-s)^2x-s^2\bigr)^2}
			{\bigl((1-s)x+s\bigr)^2}
			\ge0. \label{eq:harmonic-scalar}
		\end{align}
		Functional calculus and \eqref{eq:phi-R} therefore yield
		\[
		\altPhi_{h_s}(A,B)
		\le(1-s+s^2)\Tr A(R-I)^2,
		\]
		which proves \eqref{eq:sharp-sandwich}.
		
		The lower constant is attained identically by $f=a_s$.  For $f=h_s$,
		the pair $A=I_m,B=x^2I_m$ shows that the upper ratio reaches
		$1-s+s^2$ at $x=s^2/(1-s)^2$ when $s\ne1/2$; when $s=1/2$, the
		same ratio tends to $3/4$ as $x\to1$.  Hence the upper constant is also
		optimal.
		
	We finish with the equality statements.  Since
	\[
	a_s(R)^2-f(R)^2\succeq0
	\]
	and $A\succ0$, equality in the lower bound forces
	$a_s(R)^2=f(R)^2$.  If $A\ne B$, then
	$\sigma(R)$ contains some $r\ne1$, and positivity gives
	$f(r)=a_s(r)$.  On the interval with endpoints $1$ and $r$, the
	chord joining $(1,f(1))$ and $(r,f(r))$ is precisely $a_s$.
	Concavity places $f$ above this chord, whereas the tangent-line
	inequality gives $f\le a_s$.  Hence $f=a_s$ on that interval, and
	real analyticity yields $f=a_s$ on $(0,\infty)$.
	
	For the upper bound, equality forces equality in both steps
	\[
	\altPhi_f(A,B)
	\le \altPhi_{h_s}(A,B)
	\le (1-s+s^2)d_{\rm BW}(A,B)^2.
	\]
	Consequently,
	\[
	f(R)^2=h_s(R)^2,
	\]
	and equality holds in \eqref{eq:harmonic-scalar} throughout
	$\sigma(R)$.  Choose $r\in\sigma(R)\setminus\{1\}$.  Then
	$f(r)=h_s(r)$, or equivalently
	\[
	f^\perp(r)=a_{1-s}(r),
	\qquad
	f^\perp(x):=\frac{x}{f(x)}.
	\]
	The function $f^\perp$ is normalized operator monotone and satisfies
	$(f^\perp)'(1)=1-s$
	\cite[Section~3]{DumitruFrancoKimCzerwinska2026}.  Thus
	$f^\perp\le a_{1-s}$, and the same concavity argument shows that
	$f^\perp=a_{1-s}$ on an interval.  Real analyticity then gives
	$f^\perp=a_{1-s}$ on $(0,\infty)$, and hence $f=h_s$.
	
	The zero set in \eqref{eq:harmonic-scalar} is
	\[
	\left\{1,\frac{s^2}{(1-s)^2}\right\}.
	\]
	Therefore the asserted spectral condition is necessary.  Conversely, if
	$f=h_s$ and this spectral condition holds, functional calculus in
	\eqref{eq:harmonic-scalar} gives equality in the upper bound.  When
	$s=1/2$, the two points coincide, so no nontrivial pair attains the
	upper constant.
	\end{proof}
	
Fix a nontrivial normalized operator monotone function $f$, set
$s=f'(1)$, and define its comparison profile by
\begin{equation}\label{eq:rho-f}
	\rho_f(x):=
	\begin{cases}
		\displaystyle
		\frac{(1-s)+sx^2-f(x)^2}{(x-1)^2},
		& x\ne1,\\[2mm]
		s(1-s)-f''(1),
		& x=1.
	\end{cases}
\end{equation}
The optimal comparison constants admit the following exact description.

\begin{theorem}
	\label{thm:individual-constants}
	The best constants in
	\begin{equation}\label{eq:individual-comparison}
		c_f d_{\rm BW}(A,B)^2
		\le \altPhi_f(A,B)
		\le C_f d_{\rm BW}(A,B)^2
	\end{equation}
	valid for all positive definite matrices $A,B$ are
	\[
	c_f=\inf_{x>0}\rho_f(x),
	\qquad
	C_f=\sup_{x>0}\rho_f(x).
	\]
	Moreover,
	\[
	s(1-s)\le c_f\le C_f\le1-s+s^2.
	\]
	If $A\ne B$ and $R=A^{-1}\sharp B$, equality in the lower,
	respectively upper, bound holds if and only if
	\[
	\sigma(R)\setminus\{1\}
	\subseteq\operatorname*{argmin}_{x>0}\rho_f(x),
	\qquad\text{respectively}\qquad
	\sigma(R)\setminus\{1\}
	\subseteq\operatorname*{argmax}_{x>0}\rho_f(x).
	\]
\end{theorem}
	
\begin{proof}
	Since $f(1)=1$ and $f'(1)=s$, Taylor expansion at $x=1$
	gives
	\[
	f(x)
	=1+s(x-1)+\frac12f''(1)(x-1)^2
	+o\bigl((x-1)^2\bigr).
	\]
	Consequently,
	\[
	f(x)^2
	=1+2s(x-1)
	+\bigl(s^2+f''(1)\bigr)(x-1)^2
	+o\bigl((x-1)^2\bigr),
	\]
	whereas
	\[
	(1-s)+sx^2
	=1+2s(x-1)+s(x-1)^2.
	\]
	It follows that
	\[
	(1-s)+sx^2-f(x)^2
	=
	\bigl(s(1-s)-f''(1)\bigr)(x-1)^2
	+o\bigl((x-1)^2\bigr).
	\]
	Hence
	\[
	\lim_{x\to1}\rho_f(x)
	=s(1-s)-f''(1),
	\]
	which proves that the prescribed value at $x=1$ gives a continuous
	extension of $\rho_f$.
	
	The comparison is immediate when $A=B$, so suppose that $A\ne B$.
	Put $R=A^{-1}\sharp B$, diagonalize
	\[
	R=U\operatorname{diag}(r_1,\ldots,r_m)U^*,
	\]
	and set
	\[
	\alpha_i:=(U^*AU)_{ii}>0.
	\]
	Using \eqref{eq:bw-R}, \eqref{eq:phi-R}, and the definition of
	$\rho_f$, we obtain
	\begin{align}
		\frac{\altPhi_f(A,B)}{d_{\rm BW}(A,B)^2}
		&=
		\frac{
			\sum_{i=1}^m
			\alpha_i\bigl((1-s)+sr_i^2-f(r_i)^2\bigr)
		}{
			\sum_{i=1}^m\alpha_i(r_i-1)^2
		}\notag\\
		&=
		\frac{
			\sum_{i=1}^m
			\alpha_i(r_i-1)^2\rho_f(r_i)
		}{
			\sum_{i=1}^m\alpha_i(r_i-1)^2
		}.
		\label{eq:rho-weighted-average}
	\end{align}
	The denominator is positive because $A\ne B$, and the weights
	$\alpha_i(r_i-1)^2$ are nonnegative. Thus
	\eqref{eq:rho-weighted-average} is a weighted average of the values
	$\rho_f(r_i)$, and therefore
	\[
	\inf_{x>0}\rho_f(x)
	\le
	\frac{\altPhi_f(A,B)}{d_{\rm BW}(A,B)^2}
	\le
	\sup_{x>0}\rho_f(x).
	\]
	This proves \eqref{eq:individual-comparison} with
	\[
	c_f=\inf_{x>0}\rho_f(x),
	\qquad
	C_f=\sup_{x>0}\rho_f(x).
	\]
	
	The weights in \eqref{eq:rho-weighted-average} are positive precisely
	for those indices $i$ for which $r_i\ne1$. Hence equality in the
	lower bound holds if and only if
	\[
	\rho_f(r_i)=c_f
	\qquad\text{whenever }r_i\ne1,
	\]
	or equivalently,
	\[
	\sigma(R)\setminus\{1\}
	\subseteq\operatorname*{argmin}_{x>0}\rho_f(x).
	\]
	The same argument gives the stated characterization of equality in the
	upper bound using $\operatorname*{argmax}\rho_f$.
	
	It remains to prove that the constants cannot be improved. For any
	$x>0$, $x\ne1$, take
	\[
	A=I_m,
	\qquad
	B=x^2I_m.
	\]
	Then $A^{-1}\sharp B=xI_m$, and
	\eqref{eq:rho-weighted-average} reduces to
	\[
	\frac{\altPhi_f(A,B)}{d_{\rm BW}(A,B)^2}
	=\rho_f(x).
	\]
	Choosing sequences $x_n>0$, $x_n\ne1$, along which
	$\rho_f(x_n)$ approaches its infimum or supremum shows that neither
	constant can be improved, even in the fixed dimension $m$.
	
	Finally, Theorem~\ref{thm:sharp-comparison} shows that
	$s(1-s)$ is an admissible lower constant and $1-s+s^2$ is an
	admissible upper constant. Therefore
	\[
	s(1-s)\le c_f\le C_f\le1-s+s^2,
	\]
	completing the proof.
\end{proof}
	
	\begin{corollary}
		\label{cor:power-sharp-comparison}
		For $0<t<1$ and $f_t(x)=x^t$,
		\begin{equation}\label{eq:power-sharp-comparison}
			\min\{t,1-t\}\,d_{\rm BW}(A,B)^2
			\le \altPhi_{f_t}(A,B)
			\le \max\{t,1-t\}\,d_{\rm BW}(A,B)^2.
		\end{equation}
Both constants are optimal for the fixed power function $f_t(x)=x^t$,
	already among pairs of scalar multiples of the identity.   If $t\ne1/2$, both inequalities are
		strict whenever $A\ne B$; for $t=1/2$,
		$\altPhi_{f_{1/2}}=\tfrac12d_{\rm BW}^2$.
	\end{corollary}
	
	\begin{proof}
		It is enough to treat $0<t\le1/2$.  For $x>0$, weighted
		arithmetic--geometric mean gives
		\begin{align*}
			&(1-t)+tx^2-x^{2t}-t(x-1)^2
			=(1-2t)+2tx-x^{2t}\ge0,\\
			&(1-t)(x-1)^2-\bigl((1-t)+tx^2-x^{2t}\bigr)\\
			&\qquad=(1-2t)x^2-2(1-t)x+x^{2t}\ge0.
		\end{align*}
		For $0<t<1/2$, both inequalities are strict whenever $x\ne1$.
		For $t=1/2$, both become identities and
		$\rho_{f_{1/2}}(x)\equiv1/2$.  Optimality for the fixed power function follows by letting
		$x\to\infty$ for the lower constant and $x\downarrow0$ for the upper
		constant for the pair $(I_m,x^2I_m)$.  The case $t>1/2$ follows from
		\[
		(1-t)+tx^2-x^{2t}
		=x^2\left(t+(1-t)x^{-2}-x^{-2(1-t)}\right).
		\]
	\end{proof}

\section{Second-order Bures--Wasserstein geometry}
\label{sec:hessian}
	
For \(A\in\Pm\) and \(Y\in\Hm\), let \(\mathcal L_A(Y)\) denote the
unique Hermitian solution of the Lyapunov equation
\[
A\mathcal L_A(Y)+\mathcal L_A(Y)A=Y;
\]
see \cite[Section~4, Eqs.~(30)--(31)]{BhatiaJainLim2019}.
The Bures--Wasserstein Riemannian metric is
\begin{equation}\label{eq:bw-metric}
	g_A^{\rm BW}(Y,Z)
	:=\frac12\operatorname{Re}\Tr\!\left(\mathcal L_A(Y)Z\right);
\end{equation}
see \cite[Section~4, Eqs.~(29)--(32) and Theorem~5]
{BhatiaJainLim2019}; cf.\
\cite[Section~4, Propositions~6--7]{MalagoMontrucchioPistone2018}
for the real-symmetric formulation. If \(H=\mathcal L_A(Y)\), then
\(Y=AH+HA\), and hence
\begin{equation}\label{eq:bw-metric-diagonal}
	g_A^{\rm BW}(Y,Y)
	=\frac12\operatorname{Re}\Tr\!\bigl(H(AH+HA)\bigr)
	=\Tr(AH^2).
\end{equation}
	
	\begin{theorem}
		\label{thm:hessian}
		Let $f$ be a nontrivial normalized operator monotone function, put
		$s=f'(1)$, and
		define
		\[
		\gamma_f:=s(1-s)-f''(1).
		\]
		Then $\gamma_f>0$, and for $A\in\Pm$, $Y,Z\in\Hm$,
		\begin{equation}\label{eq:hessian}
			\left.D_B^2\altPhi_f(A,B)\right|_{B=A}[Y,Z]
			=2\gamma_f\,g_A^{\rm BW}(Y,Z).
		\end{equation}
		Equivalently,
		\[
		\altPhi_f(A,A+tY)
		=\gamma_f t^2g_A^{\rm BW}(Y,Y)+o(t^2).
		\]
	\end{theorem}
	
	\begin{proof}
		Operator monotonicity implies concavity, so $f''(1)\le0$.  Since
		$0<s<1$, this already gives $\gamma_f>0$.
		
		For small real $t$, set $B_t=A+tY$ and
		$R_t=A^{-1}\sharp B_t$.  Since the principal square-root map is real analytic on $\Pm$, the map
		\[
		t\longmapsto R_t
		=A^{-1/2}\bigl(A^{1/2}B_tA^{1/2}\bigr)^{1/2}A^{-1/2}
		\]
		is real analytic near $t=0$.  Its Taylor expansion in operator norm has
		the form
		\[
		R_t=I+tH+t^2K+O(t^3).
		\]
		Substituting this expansion into the Riccati equation
		$R_tAR_t=B_t$ and comparing the coefficients of $t$ and $t^2$
		gives
		\begin{equation}\label{eq:riccati-expansion}
			AH+HA=Y,
			\qquad
			AK+KA+HAH=0.
		\end{equation}
		Thus $H=\mathcal L_A(Y)$, and taking the trace in the second identity yields
		\begin{equation}\label{eq:trace-K}
			2\Tr(AK)=-\Tr(AH^2).
		\end{equation}
		
		Put $q(x):=f(x)^2$.  Since
		\[
		q(1)=1,\qquad q'(1)=2s,\qquad
		\frac12q''(1)=f''(1)+s^2,
		\]
		functional calculus at the scalar matrix $I$ gives
		\[
		q(R_t)
		=I+2stH
		+t^2\!\left(2sK+\bigl(f''(1)+s^2\bigr)H^2\right)
		+O(t^3).
		\]
		Using \eqref{eq:trace-K},
		\[
		\Tr(Aq(R_t))
		=\Tr A+2st\Tr(AH)
		+t^2\bigl(f''(1)+s^2-s\bigr)\Tr(AH^2)
		+O(t^3).
		\]
		On the other hand,
		\[
		(1-s)\Tr A+s\Tr B_t
		=\Tr A+st\Tr Y
		=\Tr A+2st\Tr(AH).
		\]
		The constant and linear terms therefore cancel in
		\eqref{eq:divergence}, leaving
		\[
		\altPhi_f(A,A+tY)
		=\bigl(s(1-s)-f''(1)\bigr)t^2\Tr(AH^2)+O(t^3).
		\]
		Now use \eqref{eq:bw-metric-diagonal}.  The map $t\mapsto R_t$ is real
		analytic near zero, which justifies the norm expansions above.  Finally,
		polarization gives \eqref{eq:hessian} for $Y,Z$.
	\end{proof}
	
	\begin{remark}
		If $A=\operatorname{diag}(a_1,\ldots,a_m)$, then
		\[
		\mathcal L_A(Y)_{ij}=\frac{Y_{ij}}{a_i+a_j},
		\qquad
		g_A^{\rm BW}(Y,Y)
		=\frac12\sum_{i,j=1}^m\frac{|Y_{ij}|^2}{a_i+a_j}.
		\]
		Thus the diagonal value of the Hessian is
		\[
		\left.D_B^2\altPhi_f(A,B)\right|_{B=A}[Y,Y]
		=\gamma_f\sum_{i,j=1}^m\frac{|Y_{ij}|^2}{a_i+a_j}.
		\]
	\end{remark}
	
\begin{corollary}\label{cor:quantum-divergence}
	For every nontrivial normalized operator monotone function $f$, the
	functional $\altPhi_f$ is a quantum divergence on $\Pm$ in the sense
	of Bhatia--Gaubert--Jain \cite{BGJ2019}. It is smooth and nonnegative,
	vanishes exactly on the diagonal, and has zero first derivative there in
	either variable. In fact, its second-variable Hessian at the diagonal is
	positive definite.
\end{corollary}

\begin{proof}
	Smoothness follows from analytic functional calculus.  The lower bound in
	\eqref{eq:sharp-sandwich} gives nonnegativity and shows that equality holds
	exactly when $A=B$.  For each fixed $A$, the map
	$B\mapsto\altPhi_f(A,B)$ therefore has a minimum at $B=A$, so its first
	derivative vanishes there.  The same argument, with $B$ fixed, gives
	vanishing of the first derivative in the first variable.  Finally,
	Theorem~\ref{thm:hessian} gives
	\[
	\left.D_B^2\altPhi_f(A,B)\right|_{B=A}
	=2\gamma_f g_A^{\rm BW},
	\]
	which is positive definite because $\gamma_f>0$ and
	$g_A^{\rm BW}$ is a Riemannian metric.
\end{proof}

	The coefficient $\gamma_f$ appearing in the diagonal
	Bures--Wasserstein Hessian can be estimated sharply directly from the
	Kubo--Ando representing measure.  For $\lambda\in[0,1]$, define
	\[
	h_\lambda(x)
	:=\bigl((1-\lambda)+\lambda x^{-1}\bigr)^{-1},
	\qquad x>0.
	\]
	The Kubo--Ando representation theorem gives a unique Borel measure
	$\mu_f$ on $[0,1]$ of total mass one such that
	\begin{equation}\label{eq:representing-measure}
		f(x)=\int_{[0,1]}h_\lambda(x)\,d\mu_f(\lambda),
		\qquad
		s=f'(1)=\int_{[0,1]}\lambda\,d\mu_f(\lambda).
	\end{equation}
	
	\begin{proposition}
		\label{prop:gamma-moments}
		With the notation above, assume that $f$ is nontrivial.  Then
		\begin{align}
			f''(1)
			&=-2\int_{[0,1]}
			\lambda(1-\lambda)\,d\mu_f(\lambda),
			\label{eq:f-second-moment}\\
			\gamma_f
			&=s(1-s)+2\int_{[0,1]}
			\lambda(1-\lambda)\,d\mu_f(\lambda).
			\label{eq:gamma-integral}
		\end{align}
		Consequently,
		\begin{equation}\label{eq:gamma-sharp-range}
			s(1-s)\le\gamma_f\le3s(1-s).
		\end{equation}
		Equality in the lower bound holds if and only if $f=a_s$, while
		equality in the upper bound holds if and only if $f=h_s$.
	\end{proposition}
	
	\begin{proof}
		Writing
		\[
		h_\lambda(x)
		=\frac{x}{(1-\lambda)x+\lambda},
		\]
		we obtain
		\[
		h_\lambda''(x)
		=-\frac{2\lambda(1-\lambda)}
		{\bigl((1-\lambda)x+\lambda\bigr)^3},
		\]
		and hence
		\[
		h_\lambda''(1)=-2\lambda(1-\lambda).
		\]
		The second derivatives are uniformly bounded for
		$\lambda\in[0,1]$ and $x$ in a compact neighborhood of $1$.
		We may therefore differentiate
		\eqref{eq:representing-measure} under the integral sign, obtaining
		\eqref{eq:f-second-moment}.  Since
		\[
		\gamma_f=s(1-s)-f''(1),
		\]
		equation \eqref{eq:gamma-integral} follows immediately.
		
		Now let
		\[
		r(\lambda):=\lambda(1-\lambda),
		\qquad 0\le\lambda\le1.
		\]
		The function $r$ is nonnegative and strictly concave.  Hence,
		using \eqref{eq:representing-measure} and Jensen's inequality,
		\[
		0
		\le
		\int_{[0,1]}r(\lambda)\,d\mu_f(\lambda)
		\le
		r\left(
		\int_{[0,1]}\lambda\,d\mu_f(\lambda)
		\right)
		=r(s)
		=s(1-s).
		\]
		Substitution into \eqref{eq:gamma-integral} proves
		\eqref{eq:gamma-sharp-range}.
		
		Equality in the lower bound holds if and only if
		\[
		\int_{[0,1]}\lambda(1-\lambda)\,d\mu_f(\lambda)=0.
		\]
		Since $\lambda(1-\lambda)$ is nonnegative and vanishes precisely at
		$0$ and $1$, this is equivalent to
		\[
		\operatorname{supp}\mu_f\subseteq\{0,1\}.
		\]
		The identity
		$\int_{[0,1]}\lambda\,d\mu_f(\lambda)=s$ then gives
		\[
		\mu_f=(1-s)\delta_0+s\delta_1.
		\]
		Because $h_0(x)=1$ and $h_1(x)=x$, the representation
		\eqref{eq:representing-measure} yields
		\[
		f(x)=(1-s)+sx=a_s(x).
		\]
		
		Equality in the upper bound is precisely the equality case of
		Jensen's inequality for the strictly concave function $r$.
		It holds if and only if $\mu_f$ is supported at a single point.
		The identity
		$\int_{[0,1]}\lambda\,d\mu_f(\lambda)=s$ forces that point to be
		$s$, so
		\[
		\mu_f=\delta_s.
		\]
		Equation \eqref{eq:representing-measure} then gives $f=h_s$.
	\end{proof}
	
	\section{Exact data processing and qubit pinching}
	\label{sec:dpi}
	
Let $m,\ell$ be positive integers and let
	$\Psi:\mathbb M_m\to\mathbb M_\ell$ be positive and trace preserving
	(PTP).  Since $\altPhi_f$ is defined on the open positive cone, an
	expression of the form $\altPhi_f(\Psi(A),\Psi(B))$ is understood only
	when both images are positive definite, except where we explicitly use
	the continuous extension to the closed positive cone.  If
	$\Psi(I_m)\succ0$, this image condition holds automatically for every
	$A,B\in\mathbb P_m$.
	
	We say that $\altPhi_f$ satisfies the \emph{exact data-processing
		inequality} (exact DPI) if, for every pair of positive integers $m,\ell$,
	every PTP map $\Psi:\mathbb M_m\to\mathbb M_\ell$, and every
	$A,B\in\mathbb P_m$ such that
	$\Psi(A),\Psi(B)\in\mathbb P_\ell$, one has
	\[
	\altPhi_f(\Psi(A),\Psi(B))
	\le \altPhi_f(A,B).
	\]
	Here ``exact'' means that the contraction constant is $1$, with no
	multiplicative loss; it does not refer to the equality case in the
	preceding inequality.  Since every CPTP map is PTP, this requirement is
	a priori stronger than the usual CPTP formulation.  The theorem below
	shows that the two formulations are equivalent for the family
	$\altPhi_f$.

For later use, set
\[
\mathcal F_f(A,B)
:=\Tr(A\hatsigma_fB)
=\Tr\!\bigl(Af(A^{-1}\sharp B)^2\bigr).
\]
Because a trace-preserving map preserves the two affine trace terms in
$\altPhi_f$, exact DPI is equivalent to the reverse monotonicity
\begin{equation}\label{eq:fidelity-reverse-dpi}
	\mathcal F_f(\Psi(A),\Psi(B))
	\ge \mathcal F_f(A,B).
\end{equation}
We shall also use the root fidelity, namely the homogeneous extension of
the unsquared Uhlmann fidelity to the closed positive cone:
\[
\mathcal F(A,B)
:=\Tr\!\left(A^{1/2}BA^{1/2}\right)^{1/2},
\qquad A,B\succeq0.
\]
It is monotone increasing even under PTP maps.  Indeed, by the
Bhattacharyya-coefficient characterization
\cite[Theorem~3.24]{Watrous}, there exists a finite positive
operator-valued measurement $\{M_j\}_{j=1}^N$ in $\mathbb M_\ell$ for
which equality holds in that characterization for the pair
$(\Psi(A),\Psi(B))$.

Let $\Psi^*:\mathbb M_\ell\to\mathbb M_m$ denote the adjoint map with
respect to the trace pairing.  Since $\Psi$ is positive and trace
preserving, $\Psi^*$ is positive and unital.  Hence
$\{\Psi^*(M_j)\}_{j=1}^N$ is a positive operator-valued measurement in
$\mathbb M_m$.  Applying the inequality part of the same characterization
to this pulled-back measurement gives
\begin{align*}
	\mathcal F(A,B)
	&\le
	\sum_{j=1}^N
	\sqrt{\Tr\!\bigl(A\Psi^*(M_j)\bigr)\,
		\Tr\!\bigl(B\Psi^*(M_j)\bigr)}\\
	&=
	\sum_{j=1}^N
	\sqrt{\Tr\!\bigl(\Psi(A)M_j\bigr)\,
		\Tr\!\bigl(\Psi(B)M_j\bigr)}\\
	&=\mathcal F(\Psi(A),\Psi(B)).
\end{align*}
Since
\[
d_{\rm BW}(A,B)^2
=\Tr A+\Tr B-2\mathcal F(A,B)
\]
and $\Psi$ preserves traces, it follows that
\begin{equation}\label{eq:bw-contractivity}
	d_{\rm BW}(\Psi(A),\Psi(B))
	\le d_{\rm BW}(A,B).
\end{equation}
	Let $\mathcal P:\mathbb M_2\to\mathbb M_2$ denote diagonal pinching,
	\[
	\mathcal P\!\begin{pmatrix}a&b\\c&d\end{pmatrix}
	:=\begin{pmatrix}a&0\\0&d\end{pmatrix}.
	\]
This map is PTP; in fact, it is CPTP.

The next theorem gives a complete classification of exact data processing
and shows that it is already determined by diagonal pinching on faithful
qubit states.

\begin{theorem}
	\label{thm:quadratic-square-classification}
	Let $f:(0,\infty)\to(0,\infty)$ be a normalized operator monotone
	function.  Then the following conditions are equivalent.
	\begin{enumerate}
	\item[\rm(i)] The divergence $\altPhi_f$ satisfies exact DPI under every
	positive trace-preserving map.
		
		\item[\rm(ii)] Reverse monotonicity holds for diagonal pinching on
		faithful qubit states:
		\[
		\mathcal F_f(\mathcal P(\rho),\mathcal P(\sigma))
		\ge \mathcal F_f(\rho,\sigma),
		\qquad
		\rho,\sigma\in\mathbb P_2,
		\quad
		\Tr\rho=\Tr\sigma=1.
		\]
		
		\item[\rm(iii)] The function $x\mapsto f(x)^2$ is a polynomial
		of degree at most two.
		
		\item[\rm(iv)] There exist $\alpha,\beta,\gamma\ge0$ such that
		\[
		f(x)^2=\alpha+\beta x+\gamma x^2,
		\qquad
		\alpha+\beta+\gamma=1,
		\qquad
		\beta^2\ge4\alpha\gamma.
		\]
		
		\item[\rm(v)] There exist $u,v\in[0,1]$, unique up to
		interchange, such that
		\begin{equation}\label{eq:fuv}
			f(x)=f_{u,v}(x)
			:=\sqrt{\bigl((1-u)+ux\bigr)
				\bigl((1-v)+vx\bigr)}.
		\end{equation}
		
		\item[\rm(vi)] For some positive integer $m$, there exists a
		constant $c\ge0$ such that
		\begin{equation}\label{eq:exact-bw}
			\altPhi_f(A,B)=c\,d_{\rm BW}(A,B)^2
		\end{equation}
		for all $A,B\in\mathbb P_m$.
	\end{enumerate}
These conditions remain equivalent if condition ${\rm(i)}$ is replaced
by exact DPI under all completely positive trace-preserving maps.
	
	These conditions are also equivalent to joint convexity of
	$\altPhi_f$, or equivalently joint concavity of $\mathcal F_f$,
	in every matrix size.  Moreover, validity of \eqref{eq:exact-bw}
	in one fixed matrix size forces the same identity, with the same
	constant $c$, in every matrix size.
	
	If the equivalent conditions fail, exact DPI fails already for
	diagonal pinching on a pair of real-symmetric faithful qubit states.
\end{theorem}

\begin{proof}
	Put $q:=f^2$.  Assume first that
	\[
	q(x)=\alpha+\beta x+\gamma x^2.
	\]
	Normalization gives
	\[
	\alpha+\beta+\gamma=q(1)=1.
	\]
	Since $f(0+)$ exists, $\alpha=f(0+)^2\ge0$, while positivity of
	$q$ for arbitrarily large $x$ implies $\gamma\ge0$.  The function
	$q=f^2$ is increasing, so
	\[
	q'(x)=\beta+2\gamma x\ge0.
	\]
	Letting $x\downarrow0$ gives $\beta\ge0$.  Finally, scalar
	concavity of $f$ and
	\[
	f''(x)=
	\frac{4\alpha\gamma-\beta^2}
	{4\bigl(\alpha+\beta x+\gamma x^2\bigr)^{3/2}}
	\]
	give $\beta^2\ge4\alpha\gamma$.  Thus
	${\rm(iii)}\Rightarrow{\rm(iv)}$.
	
	Suppose next that ${\rm(iv)}$ holds, and let $u,v$ be the two
	roots of
	\begin{equation}\label{eq:uv-polynomial}
		z^2-(\beta+2\gamma)z+\gamma=0.
	\end{equation}
	Its discriminant is
	\[
	(\beta+2\gamma)^2-4\gamma
	=\beta^2-4\alpha\gamma\ge0.
	\]
	Moreover,
	\[
	u+v=\beta+2\gamma\ge0,
	\qquad
	uv=\gamma\ge0,
	\]
	so $u,v\ge0$.  We also have
	\[
	(1-u)(1-v)=1-(u+v)+uv=\alpha\ge0
	\]
	and
	\[
	u+v=1-\alpha+\gamma\le2.
	\]
	These relations imply $u,v\le1$.  Hence $u,v\in[0,1]$, and
	\[
	q(x)=\bigl((1-u)+ux\bigr)\bigl((1-v)+vx\bigr).
	\]
	
	Condition ${\rm(v)}$ immediately implies ${\rm(iii)}$.  Moreover,
	for $u,v\in[0,1]$ and $X\succ0$, the matrices
	$(1-u)I+uX$ and $(1-v)I+vX$ commute, and therefore
	\[
	f_{u,v}(X)
	=\bigl((1-u)I+uX\bigr)
	\sharp
	\bigl((1-v)I+vX\bigr).
	\]
	Both affine maps are operator monotone, and the geometric mean is
	jointly monotone.  Thus $f_{u,v}$ is a normalized operator monotone
	function.  This proves the equivalence of
	${\rm(iii)}$--${\rm(v)}$, and uniqueness up to interchange follows
	from the sum and product in \eqref{eq:uv-polynomial}.
	
For $f=f_{u,v}$, set
\[
\alpha=(1-u)(1-v),
\qquad
\beta=u+v-2uv,
\qquad
\gamma=uv,
\]
and put $s:=f'(1)$.  Since $q=f^2$ and $f(1)=1$,
\[
2s=q'(1)=u+v.
\]
Thus
\[
s=\frac{u+v}{2}.
\]
Consequently,
\[
(1-s)-\alpha=s-\gamma=\frac{\beta}{2}.
\]
Hence
	\[
	(1-s)+sx^2-f(x)^2=\frac{\beta}{2}(x-1)^2.
	\]
	Functional calculus and
	Proposition~\ref{prop:master-representation} give
	\eqref{eq:exact-bw} in every matrix size, with
	$c=\beta/2$.  This proves
	${\rm(v)}\Rightarrow{\rm(vi)}$.
	
	Conversely, assume ${\rm(vi)}$ and set
	\[
	A=I_m,
	\qquad
	B=x^2I_m.
	\]
	After cancelling the common factor $m$, identity
	\eqref{eq:exact-bw} becomes
	\[
	(1-s)+sx^2-f(x)^2=c(x-1)^2.
	\]
	Thus $f^2$ is a polynomial of degree at most two.  Hence
	${\rm(vi)}\Rightarrow{\rm(iii)}$, and
	${\rm(iii)}$--${\rm(vi)}$ are equivalent.  The preceding
	calculation also shows that \eqref{eq:exact-bw} then holds, with the
	same $c$, in every matrix size.
	
By \eqref{eq:exact-bw} and the contractivity
\eqref{eq:bw-contractivity}, condition ${\rm(iii)}$ implies
${\rm(i)}$ for every PTP map.  Since diagonal pinching is PTP,
\eqref{eq:fidelity-reverse-dpi} also shows that
${\rm(i)}$ implies ${\rm(ii)}$.
	
	It remains to prove ${\rm(ii)}\Rightarrow{\rm(iii)}$.  Fix
$0<x<1<y$, and put
\[
\Lambda=\begin{pmatrix}x&0\\0&y\end{pmatrix},
\qquad
K=\begin{pmatrix}0&1\\1&0\end{pmatrix},
\]
and
\[
a=\frac{y^2-1}{y^2-x^2},
\qquad
b=\frac{1-x^2}{y^2-x^2},
\qquad
A_0=\begin{pmatrix}a&0\\0&b\end{pmatrix}.
\]
Then $a,b>0$, $a+b=1$, and $ax^2+by^2=1$.  Set
\[
\eta=ax+by.
\]
For sufficiently small real $r,t$, consider
\[
A_r=A_0+rK,
\qquad
B_{r,t}=\Lambda A_r\Lambda+t\eta K.
\]
Since $K$ and $\Lambda K\Lambda=xyK$ are traceless, we have
\[
\Tr A_r
=\Tr A_0+r\Tr K
=a+b=1,
\]
and
\[
\Tr B_{r,t}
=\Tr(\Lambda A_0\Lambda)+r\Tr(\Lambda K\Lambda)+t\eta\Tr K
=ax^2+by^2=1.
\]
Since $A_0>0$ and $\Lambda A_0\Lambda>0$, the matrices $A_r$ and
$B_{r,t}$ remain positive definite for sufficiently small $r,t$.
Together with the trace identities above, this shows that they are
faithful states.  Moreover, since $K$ and $\Lambda K\Lambda$ are
off-diagonal, diagonal pinching gives
\[
\mathcal P(A_r)=A_0,
\qquad
\mathcal P(B_{r,t})=\Lambda A_0\Lambda.
\]
By ${\rm(ii)}$,
\[
G(r,t):=\mathcal F_f(A_r,B_{r,t})
\le \mathcal F_f(A_0,\Lambda A_0\Lambda)=G(0,0),
\]
so $G$ has a local maximum at the origin.

Let
\[
R(r,t)=A_r^{-1}\sharp B_{r,t}.
\]
Then
\[
R(r,t)A_rR(r,t)=B_{r,t}.
\]
When $t=0$, we have $B_{r,0}=\Lambda A_r\Lambda$.  Since
$\Lambda>0$ and $R(r,0)$ is the unique positive definite solution of
\[
X A_r X=B_{r,0},
\]
it follows that $R(r,0)=\Lambda$.  Therefore, defining
$g(r):=G(r,0)$ and using
$q(\Lambda)=\operatorname{diag}(q(x),q(y))$, we obtain
\[
\begin{aligned}
	g(r)
	&=\Tr\!\left(A_rq(\Lambda)\right)\\
	&=\Tr\!\begin{pmatrix}
		a q(x)&r q(y)\\
		r q(x)&b q(y)
	\end{pmatrix}\\
	&=a q(x)+b q(y).
\end{aligned}
\]
Thus $g$ is constant and, in particular, $g''(0)=0$.  Since $G$
has a local maximum at the origin, its Hessian at $(0,0)$ is negative
semidefinite.  Its $(r,r)$-entry is
\[
\frac{\partial^2G}{\partial r^2}(0,0)=g''(0)=0.
\]
Thus
\[
\operatorname{Hess}G(0,0)
=
\begin{pmatrix}
	0&G_{rt}(0,0)\\
	G_{rt}(0,0)&G_{tt}(0,0)
\end{pmatrix}
\preceq0.
\]
Taking determinants gives
\[
0\le\det\!\bigl(\operatorname{Hess}G(0,0)\bigr)
=-G_{rt}(0,0)^2.
\]
Therefore
\begin{equation}\label{eq:mixed-pinching-zero}
	G_{rt}(0,0)
	=\frac{\partial^2G}{\partial r\,\partial t}(0,0)=0.
\end{equation}
We now obtain a second expression for the same mixed derivative by
differentiating the Riccati equation
\[
R(r,t)A_rR(r,t)=B_{r,t}
\]
with respect to $t$. Since $A_r$ is independent of $t$ and
$\partial B_{r,t}/\partial t=\eta K$, the product rule gives
\[
\frac{\partial R}{\partial t}(r,t)A_rR(r,t)
+R(r,t)A_r\frac{\partial R}{\partial t}(r,t)
=\eta K.
\]
Along $t=0$, we have $R(r,0)=\Lambda$, and hence
\[
\frac{\partial R}{\partial t}(r,0)A_r\Lambda
+\Lambda A_r\frac{\partial R}{\partial t}(r,0)
=\eta K.
\]
In particular, at $r=0$,
\[
\frac{\partial R}{\partial t}(0,0)A_0\Lambda
+\Lambda A_0\frac{\partial R}{\partial t}(0,0)
=\eta K.
\]

We use that, for $H\succ0$, the Lyapunov map
$Z\mapsto ZH+HZ$ is invertible: in an eigenbasis of $H$, its
$(i,j)$-entry is multiplication by $h_i+h_j>0$.

At $r=0$, put
\[
X:=\frac{\partial R}{\partial t}(0,0),
\qquad
C:=A_0\Lambda=\Lambda A_0=\operatorname{diag}(ax,by)>0.
\]
The preceding equation becomes
\[
XC+CX=\eta K.
\]
Since $\eta=ax+by$, we have
\[
KC+CK=(ax+by)K=\eta K.
\]
Hence,
\[
\frac{\partial R}{\partial t}(0,0)=X=K.
\]
Next, differentiate
\[
\frac{\partial R}{\partial t}(r,0)A_r\Lambda
+\Lambda A_r\frac{\partial R}{\partial t}(r,0)=\eta K
\]
with respect to $r$.  Since
$\frac{dA_r}{dr}=K$
and the right-hand side is independent of $r$, the product rule gives
\[
\begin{aligned}
	0={}&
	\frac{\partial^2R}{\partial r\,\partial t}(r,0)A_r\Lambda
	+\frac{\partial R}{\partial t}(r,0)K\Lambda\\
	&\quad
	+\Lambda K\frac{\partial R}{\partial t}(r,0)
	+\Lambda A_r\frac{\partial^2R}{\partial r\,\partial t}(r,0).
\end{aligned}
\]
Evaluating at $r=0$ and using
$\frac{\partial R}{\partial t}(0,0)=K$, we obtain, with
\[
Y:=\frac{\partial^2R}{\partial r\,\partial t}(0,0),
\]
that
\[
YA_0\Lambda+KK\Lambda+\Lambda KK+\Lambda A_0Y=0.
\]
Since $K^2=I$ and $A_0\Lambda=\Lambda A_0=C$, this reduces to
\[
YC+CY=-2\Lambda.
\]
Moreover,
\[
(-A_0^{-1})C+C(-A_0^{-1})=-2\Lambda.
\]
We get
\[
\frac{\partial^2R}{\partial r\,\partial t}(0,0)
=Y=-A_0^{-1}.
\]
Since $R_r(0,0)=0$, the chain rule and the divided-difference formula
\cite[Chapter~3]{Higham2008}
\[
{\rm D}q(\Lambda)[K]
=\frac{q(x)-q(y)}{x-y}\,K
\]
give
\begin{align*}
	G_{rt}(0,0)
	&=\Tr\!\left(K\,{\rm D}q(\Lambda)[K]\right)
	+\Tr\!\left(A_0{\rm D}q(\Lambda)[-A_0^{-1}]\right)\\
	&=2\frac{q(x)-q(y)}{x-y}-q'(x)-q'(y).
\end{align*}
Together with \eqref{eq:mixed-pinching-zero}, this proves
\begin{equation}\label{eq:trapezoidal-q}
	2\frac{q(x)-q(y)}{x-y}=q'(x)+q'(y),
	\qquad 0<x<1<y.
\end{equation}
Differentiating \eqref{eq:trapezoidal-q} with respect to each endpoint
gives
\[
q''(x)=\frac{q'(x)-q'(y)}{x-y}=q''(y).
\]
Fixing one endpoint at a time shows that $q''$ is constant on
$(0,1)$ and on $(1,\infty)$; analyticity gives the same value at
$1$.  Thus $q=f^2$ is a polynomial of degree at most two.  This
completes the equivalence of ${\rm(i)}$--${\rm(vi)}$.  It remains to justify the CPTP formulation stated in the theorem. Let
$({\rm i}_{\rm CPTP})$ denote exact DPI under every completely positive
trace-preserving map. Since diagonal pinching is CPTP,
$({\rm i}_{\rm CPTP})$ implies ${\rm(ii)}$. Conversely,
${\rm(iii)}$ implies ${\rm(i)}$, which is exact DPI under every PTP
map and hence, in particular, under every CPTP map. Therefore
$({\rm i}_{\rm CPTP})$ is equivalent to conditions
${\rm(i)}$--${\rm(vi)}$.

For the convexity--concavity assertion, under ${\rm(iii)}$ we have
\[
\mathcal F_f(A,B)
=\alpha\Tr A+\beta\mathcal F(A,B)+\gamma\Tr B.
\]
The root fidelity $\mathcal F$ is jointly concave and $\beta\ge0$,
so $\mathcal F_f$ is jointly concave and $\altPhi_f$ is jointly
convex.  Conversely, either property in size two, together with
simultaneous unitary invariance, implies ${\rm(ii)}$, because
\[
\mathcal P(X)=\frac12(X+UXU^*),
\qquad
U=\operatorname{diag}(1,-1).
\]

Finally, suppose that $q$ is not a polynomial of degree at most two.
Then \eqref{eq:trapezoidal-q} fails for some $0<x<1<y$.  For the
corresponding $G$, note that
\[
G(r,0)=\Tr(A_0q(\Lambda)).
\]
Hence
\[
G_r(0,0)=G_{rr}(0,0)=0.
\]
Moreover, the chain rule and $R_t(0,0)=K$ give
\[
G_t(0,0)
=\Tr\!\left(A_0{\rm D}q(\Lambda)[K]\right)=0,
\]
since
\[
{\rm D}q(\Lambda)[K]
=\frac{q(x)-q(y)}{x-y}K
\]
is off-diagonal.  Thus $\nabla G(0,0)=0$.

By the choice of $x$ and $y$, the trapezoidal identity fails, and the
preceding mixed-derivative formula gives
\[
G_{rt}(0,0)
=
2\frac{q(x)-q(y)}{x-y}-q'(x)-q'(y)
\ne0.
\]
Consequently,
\[
\operatorname{Hess}G(0,0)
=
\begin{pmatrix}
	0&G_{rt}(0,0)\\
	G_{rt}(0,0)&G_{tt}(0,0)
\end{pmatrix},
\qquad
\det\!\bigl(\operatorname{Hess}G(0,0)\bigr)
=-G_{rt}(0,0)^2<0.
\]
Thus $\operatorname{Hess}G(0,0)$ is indefinite. Choose
$v\in\mathbb R^2$ such that
\[
v^{T}\operatorname{Hess}G(0,0)v>0.
\]
Since $\nabla G(0,0)=0$, Taylor's formula yields
\[
G(\varepsilon v)
=
G(0,0)
+\frac{\varepsilon^2}{2}
v^{T}\operatorname{Hess}G(0,0)v
+o(\varepsilon^2)
>G(0,0)
\]
for all sufficiently small positive $\varepsilon$.  Along this
direction, the corresponding matrices $A_r$ and $B_{r,t}$ remain
real-symmetric faithful states, while
\[
\mathcal P(A_r)=A_0,
\qquad
\mathcal P(B_{r,t})=\Lambda A_0\Lambda.
\]
Hence
\[
\mathcal F_f(A_r,B_{r,t})
>
\mathcal F_f(\mathcal P(A_r),\mathcal P(B_{r,t})),
\]
which is the asserted strict faithful qubit pinching counterexample.
\end{proof}

\begin{corollary}
	\label{cor:exact-dpi-parameters}
	Suppose that the equivalent conditions of
	Theorem~\ref{thm:quadratic-square-classification} hold, and write
	$f=f_{u,v}$ as in \eqref{eq:fuv}.  Then
	\begin{equation}\label{eq:fuv-coefficients}
		\alpha=(1-u)(1-v),
		\qquad
		\beta=u+v-2uv,
		\qquad
		\gamma=uv,
		\qquad
		c=\frac{\beta}{2},
	\end{equation}
	and
	\[
	f'(1)=\frac{u+v}{2}.
	\]
	In particular,
	\[
	\altPhi_f=\frac{\beta}{2}\,d_{\rm BW}^2
	\]
	in every matrix size.
	
	At fixed $s=f'(1)$, the exact-DPI class is parameterized by
	\begin{equation}\label{eq:fixed-weight-exact-dpi}
		f(x)^2
		=(1-2s+\gamma)+2(s-\gamma)x+\gamma x^2,
		\qquad
		\max\{0,2s-1\}\le\gamma\le s^2,
	\end{equation}
	and then
	\[
	\altPhi_f=(s-\gamma)d_{\rm BW}^2.
	\]
	
	Moreover, $f_{u,v}$ is nontrivial if and only if
	\[
	(u,v)\notin\{(0,0),(1,1)\}.
	\]
	In every nontrivial case, $\beta>0$.
\end{corollary}

\begin{proof}
	Expansion of \eqref{eq:fuv} gives
	\eqref{eq:fuv-coefficients}, except for the identity involving $c$.
	Since $2f'(1)=(f^2)'(1)=u+v$, we have
	$s=(u+v)/2$, and hence
	\[
	(1-s)-\alpha=s-\gamma=\frac{\beta}{2}.
	\]
	Proposition~\ref{prop:master-representation} now gives
	$c=\beta/2$ and the asserted Bures--Wasserstein identity.
	
	For fixed $s$, one has $u+v=2s$ and $\gamma=uv$.  As
	$u,v\in[0,1]$ vary with this fixed sum, their product ranges exactly
	over
	\[
	\max\{0,2s-1\}\le\gamma\le s^2.
	\]
	Substitution of
	\[
	\alpha=1-2s+\gamma,
	\qquad
	\beta=2(s-\gamma)
	\]
	gives \eqref{eq:fixed-weight-exact-dpi} and the corresponding
	Bures--Wasserstein factor.
	
	Finally,
	\[
	\beta=u(1-v)+v(1-u),
	\]
	which vanishes exactly at $(u,v)=(0,0)$ and $(1,1)$.  In every
	other case $\beta>0$.
\end{proof}

Identity \eqref{eq:exact-bw} defines a continuous extension of
$\altPhi_f$ to positive semidefinite pairs whenever the equivalent
conditions of Theorem~\ref{thm:quadratic-square-classification} hold.
For this extension, exact DPI under positive trace-preserving maps remains
valid on the closed cone.

Two notable functions satisfying the exact-DPI property are worth
mentioning.  The
choice $(u,v)=(0,1)$ gives $f(x)=\sqrt{x}$ and
$\altPhi_f=\tfrac12d_{\rm BW}^2$, while the diagonal choice $u=v=s$
gives $f(x)=(1-s)+sx$ and
$\altPhi_f=s(1-s)d_{\rm BW}^2$.

For the power family $f_t(x)=x^t$, where $0<t<1$,
	Theorem~\ref{thm:quadratic-square-classification} itself shows that exact
	DPI holds if and only if $t=1/2$: indeed, $f_t(x)^2=x^{2t}$ is a
	polynomial of degree at most two if and only if $2t=1$.  When $t\ne1/2$,
	the last assertion of that theorem gives a counterexample under diagonal
	pinching on real-symmetric faithful qubit states.  This conclusion agrees
	with \cite[Theorem~5.4]{LeVuHoDinh2026}.
	Positive-definite numerical pinching examples, including one at
	$t=0.3$, also appear in \cite[Proposition~4.3]{VuongDinh2026}.

	The sharp comparison also yields a dimension-free approximate form of data
	processing.

	\begin{corollary}
		\label{cor:relaxed-dpi}
 Let $f:(0,\infty)\to(0,\infty)$ be a nontrivial normalized operator
	monotone function, let $s=f'(1)$, and let $c_f,C_f$ be the exact
	constants from Theorem~\ref{thm:individual-constants}.  Let
$\Psi:\mathbb M_m\to\mathbb M_\ell$ be positive and trace preserving.
If
\[
A,B\in\mathbb P_m,
\qquad
\Psi(A),\Psi(B)\in\mathbb P_\ell,
\]
then
		\begin{equation}\label{eq:relaxed-dpi}
			\altPhi_f(\Psi(A),\Psi(B))
			\le \frac{C_f}{c_f}\,\altPhi_f(A,B)
			\le \frac{1-s+s^2}{s(1-s)}\,\altPhi_f(A,B).
		\end{equation}
	\end{corollary}

	\begin{proof}
		This follows from Theorem~\ref{thm:individual-constants},
		contractivity \eqref{eq:bw-contractivity}, and
		$c_f\ge s(1-s)$, $C_f\le1-s+s^2$.
	\end{proof}

We close the section with a tensor-rigidity consequence.  For
$f_t(x)=x^t$, the restriction of $\mathcal F_{f_t}$ to density operators
is precisely the weighted spectral fidelity considered in
\cite[Proposition~3.4]{LeVuHoDinh2026}, where its tensor multiplicativity
was proved.  The following proposition extends this multiplicativity to
arbitrary positive definite inputs and establishes the converse among all
normalized operator monotone functions.

	\begin{proposition}
		\label{prop:tensor-rigidity}
		The identity
		\begin{equation}\label{eq:tensor-multiplicativity}
			\mathcal F_f(A_1\otimes A_2,B_1\otimes B_2)
			=\mathcal F_f(A_1,B_1)\mathcal F_f(A_2,B_2)
		\end{equation}
		for all positive definite pairs of arbitrary sizes holds if and only if
			$f(x)=x^t$ for some $t\in[0,1]$.  For any nontrivial tensor-multiplicative normalized operator monotone
			function $f$, exact DPI for $\altPhi_f$, symmetry of $\mathcal F_f$,
			and the identity $f(x)=\sqrt{x}$ are equivalent.
	\end{proposition}

	\begin{proof}
		If $f(x)=x^t$ and $R_j=A_j^{-1}\sharp B_j$, the Riccati
		characterization gives
		\[
		(A_1\otimes A_2)^{-1}\sharp(B_1\otimes B_2)=R_1\otimes R_2,
		\]
		so functional calculus and factorization of the trace prove
		\eqref{eq:tensor-multiplicativity}.  Conversely, applying that identity
		to the scalar pairs $(1,x^2)$ and $(1,y^2)$ gives
		$f(xy)^2=f(x)^2f(y)^2$.  Positivity therefore gives
		$f(xy)=f(x)f(y)$.  Continuity, monotonicity, and concavity then yield
		$f(x)=x^t$ with $t\in[0,1]$.

		In the nontrivial case, $0<t<1$. By
		Theorem~\ref{thm:quadratic-square-classification}, exact DPI implies that
		\[
		f(x)^2=x^{2t}
		\]
		is a polynomial of degree at most two. Since $0<2t<2$, this is possible
		only when $2t=1$, and hence $t=1/2$.
		Symmetry applied to the scalar pair $(1,x^2)$ likewise gives
		$x^{2t}=x^{2(1-t)}$, hence $t=1/2$.  The converse follows from the
		standard properties of Uhlmann root fidelity and
		$\altPhi_{f_{1/2}}=\tfrac12d_{\rm BW}^2$.
	\end{proof}

	\section{Barycenters}
	\label{sec:barycenters}
	
Throughout this section,
	$f:(0,\infty)\to(0,\infty)$ is a nontrivial normalized operator
	monotone function, $q(x):=f(x)^2$, and
	$s:=f'(1)\in(0,1)$.  For $A\in\Pm$ and
	$X\succeq0$, put
	\begin{equation}\label{eq:RA}
		C_A(X):=A^{1/2}XA^{1/2},
		\qquad
		R_A(X):=A^{-1/2}C_A(X)^{1/2}A^{-1/2}.
	\end{equation}
	For $X\succ0$, this is $A^{-1}\sharp X$, and
	\begin{equation}\label{eq:RA-properties}
		R_A(X)AR_A(X)=X,
		\qquad
		\Tr\!\left(AR_A(X)^2\right)=\Tr X.
	\end{equation}
	
Fix $A_1,\ldots,A_n\in\Pm$ and a positive probability vector
$\omega=(w_1,\ldots,w_n)$.
	Define
	\begin{align}
		H_{f,A}(X)
		&:=\Tr\!\left(Aq(R_A(X))\right), \label{eq:HfA}\\
		H_{f,\omega}(X)
		&:=\sum_{j=1}^n w_jH_{f,A_j}(X), \label{eq:Hfomega}\\
		\mathcal E_f(X)
		&:=\sum_{j=1}^n w_j\altPhi_f(A_j,X) \notag\\
		&=C_\omega+s\Tr X-H_{f,\omega}(X), \label{eq:energy}
	\end{align}
	where
	\[
	C_\omega:=(1-s)\sum_{j=1}^n w_j\Tr A_j.
	\]
Since $f(0+)$ exists and is finite, all these functions extend
	continuously to the closed cone
	\[
	\overline{\mathbb{P}}_m:=\{X\in\Hm:X\succeq0\}.
	\]
	We use the following standard optimization terminology.  If
	$F:\mathcal D\to\mathbb R$ is a real-valued function on a feasible set
	$\mathcal D\subseteq\Hm$, then $X_*\in\mathcal D$ is a
	\emph{global minimizer} if
	$F(X_*)\le F(X)$ for every $X\in\mathcal D$, and a \emph{local
		minimizer} if the same inequality holds for every $X\in\mathcal D$
	sufficiently close to $X_*$.  The set of global minimizers is denoted by
	$\operatorname*{argmin}_{X\in\mathcal D}F(X)$.  If $\mathcal D$ is open
	and $F$ is differentiable, then $X_*$ is a \emph{critical point} if
	$DF(X_*)[Y]=0$ for every direction $Y\in\Hm$; the directional derivative
	$DF(X_*)[Y]$ is also called the first variation of $F$ at $X_*$.  The
	function $F$ is called \emph{coercive} on $\mathcal D$ if
	$F(X_k)\to+\infty$ whenever $X_k\in\mathcal D$ and
	$\|X_k\|_{\rm HS}\to\infty$.  A \emph{closed-cone minimizer} of
	$\mathcal E_f$ means a global minimizer over $\overline{\mathbb{P}}_m$.

	\subsection{Closed-cone existence, first variation, and scalar uniqueness}
	
We first establish the existence of a minimizer on the closed cone
	$\overline{\mathbb{P}}_m$, without any additional assumption on $f$.

\begin{proposition}
		\label{prop:closed-cone}
		For every nontrivial normalized operator monotone function $f$, the set
		\[
		\operatorname*{argmin}_{X\succeq0}\mathcal E_f(X)
		\]
		is nonempty and compact.  Moreover,
		\begin{equation}\label{eq:open-closed}
			\inf_{X\in\Pm}\mathcal E_f(X)
			=\min_{X\succeq0}\mathcal E_f(X),
		\end{equation}
		and $X=0$ is not a minimizer.
	\end{proposition}
	
	\begin{proof}
		Put
		\[
		c_-:=s(1-s),\qquad
		T_\omega:=\sum_{j=1}^n w_j\Tr A_j,\qquad
		M_\omega:=\sum_{j=1}^n w_j\sqrt{\Tr A_j}.
		\]
		The lower bound in \eqref{eq:sharp-sandwich} extends continuously to
		$X\succeq0$. Since
		\[
		\Tr(A^{1/2}XA^{1/2})^{1/2}
		=\|X^{1/2}A^{1/2}\|_1
		\le \|X^{1/2}\|_2\|A^{1/2}\|_2
		=\sqrt{\Tr X\,\Tr A},
		\]
	      we obtain
		\begin{eqnarray} \label{eq:coercivity}
				\mathcal E_f(X)
			&\ge& c_-\bigl(\Tr X-2M_\omega\sqrt{\Tr X}+T_\omega\bigr)\\ \nonumber
			&=&c_- \Big( (\sqrt{\Tr X}-M_\omega)^2 +T_{\omega} -M^2_{\omega}.   \Big)
			\end{eqnarray}
		
		Note that
		\begin{eqnarray*}
			M_{\omega}^2 = \Big( \sum_{j=1}^n \sqrt{\omega_j}. \sqrt{\omega_j \Tr A_j} \Big)^2 \le \sum_{j=1}^n \omega_j \sum_{j=1}^n \omega_j \Tr A_j=T_{\omega}.
			\end{eqnarray*}
Thus $\mathcal{E}_f(X) \rightarrow \infty$ as $\operatorname{Tr} X \rightarrow \infty$. 
		Hence, for every $a\in\mathbb R$, the sublevel set
		\[
		\{X\in\overline{\mathbb{P}}_m:\mathcal E_f(X)\le a\}
		\]
		is bounded: indeed, for $X\succeq0$, bounded trace bounds every
		eigenvalue of $X$ and hence bounds $X$ itself.  The sublevel set is also
		closed by continuity, and therefore compact in the finite-dimensional
		space $\Hm$.  Taking
		$a=\mathcal E_f(I)$ and applying continuity on this nonempty compact set
		shows that $\mathcal E_f$ attains its minimum on $\overline{\mathbb{P}}_m$.
		Its minimizer set is a closed subset of a compact sublevel set, and is
		therefore compact.  If $X_0\succeq0$ is a minimizer, then
		$X_0+\varepsilon I\in\Pm$ and, by continuity,
		\[
		\mathcal E_f(X_0+\varepsilon I)
		\longrightarrow \mathcal E_f(X_0)
		\qquad (\varepsilon\downarrow0).
		\]
		This proves \eqref{eq:open-closed}.
		
	It remains to show that $X=0$ is not a minimizer. Set $a_0=f(0+)$.  Nontriviality gives
		$0\le a_0<1$: if $a_0=1$, monotonicity and $f(1)=1$ force
		$f\equiv1$ on $(0,1]$, hence everywhere by analyticity.  For
		$\varepsilon>0$,
		\[
		R_j:=R_{A_j}(\varepsilon I)=\sqrt{\varepsilon}\,A_j^{-1/2}.
		\]
		Choose $\varepsilon$ so small that
		$\sqrt{\varepsilon}\,\|A_j^{-1/2}\|\le1$ for every $j$.
	If $a_0>0$, concavity and normalization give
		\[
		f(r)\ge a_0+(1-a_0)r,\qquad 0\le r\le1.
		\]
		Consequently,
		\[
		f(r)^2-a_0^2
		\ge2a_0(1-a_0)r+(1-a_0)^2r^2
		\ge2a_0(1-a_0)r.
		\]
		Our choice of $\varepsilon$ ensures that $0\preceq R_j\preceq I$.
		Functional calculus therefore yields
		\[
		f(R_j)^2-a_0^2I
		\succeq2a_0(1-a_0)R_j.
		\]
		Since $R_{A_j}(0)=0$, we have
		$H_{f,A_j}(0)=a_0^2\Tr A_j$, and hence
		\begin{align*}
			H_{f,A_j}(\varepsilon I)-H_{f,A_j}(0)
			&=\Tr\!\left(
			A_j\bigl(f(R_j)^2-a_0^2I\bigr)
			\right)\\
			&\ge2a_0(1-a_0)\Tr(A_jR_j)\\
			&=2a_0(1-a_0)\sqrt{\varepsilon}\,
			\Tr A_j^{1/2}.
		\end{align*}
		Summing over $j$ gives
		\[
		H_{f,\omega}(\varepsilon I)-H_{f,\omega}(0)
		\ge2a_0(1-a_0)\sqrt{\varepsilon}
		\sum_{j=1}^n w_j\Tr A_j^{1/2}.
		\]  
		The coefficient
		\[
		C:=2a_0(1-a_0)\sum_{j=1}^n w_j\Tr A_j^{1/2}
		\]
		is strictly positive.  Since
		\[
		\frac{C\sqrt{\varepsilon}}{sm\varepsilon}
		=\frac{C}{sm\sqrt{\varepsilon}}
		\longrightarrow+\infty
		\qquad(\varepsilon\downarrow0),
		\]
		we have
		\[
		H_{f,\omega}(\varepsilon I)-H_{f,\omega}(0)
		>sm\varepsilon
		\]
		for all sufficiently small $\varepsilon>0$.
		
		If $a_0=0$, then the continuous extension of $f$ satisfies
		$f(0)=0$ and $f(1)=1$.  By concavity, for $0\le r\le1$,
		\[
		f(r)
		=f\bigl((1-r)0+r1\bigr)
		\ge(1-r)f(0)+rf(1)
		=r.
		\]
		Hence functional calculus gives $f(R_j)\succeq R_j$.  Since both
		matrices are functions of $R_j$, they commute, and therefore
		$f(R_j)^2\succeq R_j^2$.  Hence,
		\begin{align*}
			H_{f,A_j}(\varepsilon I)-H_{f,A_j}(0)
			&=\Tr\!\left(A_jf(R_j)^2\right)\\
			&\ge\Tr(A_jR_j^2)
			=\Tr(\varepsilon I)
			=m\varepsilon.
		\end{align*}
Therefore,
		\[
		H_{f,\omega}(\varepsilon I)-H_{f,\omega}(0)
		\ge m\varepsilon>sm\varepsilon,
		\]
		because $0<s<1$. Thus, in both cases,
		\[
		H_{f,\omega}(\varepsilon I)-H_{f,\omega}(0)
		>sm\varepsilon
		\]
		for all sufficiently small $\varepsilon>0$.  In either case, \eqref{eq:energy} gives
		\begin{align*}
			\mathcal E_f(\varepsilon I)-\mathcal E_f(0)
			=sm\varepsilon
			-\bigl(
			H_{f,\omega}(\varepsilon I)-H_{f,\omega}(0)
			\bigr)<0.
		\end{align*}
		Consequently,
		$\mathcal E_f(\varepsilon I)<\mathcal E_f(0)$ for all sufficiently
		small $\varepsilon>0$, and hence $X=0$ is not a minimizer.

	\end{proof}
	
	For $C\in\Pm$, denote the Fr\'echet derivative of the square-root map by
	\[
	\mathcal S_C[Z]:=D(C^{1/2})[Z]
	=\frac1\pi\int_0^\infty
	\lambda^{1/2}(\lambda I+C)^{-1}
	Z(\lambda I+C)^{-1}\,d\lambda.
	\]
	It is positive and self-adjoint for the Hilbert--Schmidt pairing; see, for
	example, \cite[Chapter~6]{Higham2008}.  More generally, for a differentiable
	real function $g$ on $(0,\infty)$, define
	\begin{equation}\label{eq:Kg}
		\mathcal K_g(A;X)
		:=A^{1/2}\mathcal S_{C_A(X)}\!\left[
		A^{-1/2}Dg(R_A(X))[A]A^{-1/2}
		\right]A^{1/2},
		\qquad X\in\Pm.
	\end{equation}
	
 \begin{proposition}
		\label{prop:euler}
		Under the standing assumptions, recall that $q(x)=f(x)^2$ and
		$s=f'(1)\in(0,1)$.  For $A,X\in\Pm$ and $Y\in\Hm$,
		\begin{equation}\label{eq:first-variation}
			D_X\altPhi_f(A,X)[Y]
			=\Tr\!\left((sI-\mathcal K_q(A;X))Y\right).
		\end{equation}
		Hence every local minimizer $X_*\in\Pm$ satisfies the barycenter equation
		\begin{equation}\label{eq:euler}
			sI=\sum_{j=1}^n w_j\mathcal K_q(A_j;X_*).
		\end{equation}
		Conversely, if $\mathcal E_f$ is convex on $\Pm$, every solution of
		\eqref{eq:euler} is a global minimizer over $\Pm$.  If $\mathcal E_f$ is
		strictly convex on $\Pm$, then \eqref{eq:euler} has at most one solution;
		whenever such a solution exists, it is the unique global minimizer over
		$\Pm$.
\end{proposition}

 \begin{proof}
		Fix $A,X\in\Pm$ and $Y\in\Hm$, and put
		\[
		C:=C_A(X),\qquad R:=R_A(X).
		\]
		Since
		\[
		DC_A(X)[Y]=A^{1/2}YA^{1/2},
		\]
		the chain rule gives
		\[
		DR_A(X)[Y]
		=A^{-1/2}\mathcal S_C
		[A^{1/2}YA^{1/2}]A^{-1/2}.
		\]
		
		By the chain rule and the self-adjointness of $Dq(R)$ for the
		Hilbert--Schmidt pairing,
		\begin{align*}
			DH_{f,A}(X)[Y]
			&=\Tr\!\left(A\,Dq(R)[DR_A(X)[Y]]\right)\\
			&=\Tr\!\left(
			A^{-1/2}Dq(R)[A]A^{-1/2}
			\mathcal S_C[A^{1/2}YA^{1/2}]
			\right).
		\end{align*}
		Using the self-adjointness of $\mathcal S_C$ and cyclicity of the
		trace, we obtain
		\begin{align*}
			DH_{f,A}(X)[Y]
			&=\Tr\!\left(
			A^{1/2}\mathcal S_C\!\left[
			A^{-1/2}Dq(R)[A]A^{-1/2}
			\right]A^{1/2}Y
			\right)\\
			&=\Tr\!\left(\mathcal K_q(A;X)Y\right).
		\end{align*}
		Here the self-adjointness of $Dq(R)$ follows from the
		Daleckii--Krein formula; see \cite[Chapter~3]{Higham2008}.
		
		Since
		\[
		\altPhi_f(A,X)
		=(1-s)\Tr A+s\Tr X-H_{f,A}(X),
		\]
		it follows that
		\[
		D_X\altPhi_f(A,X)[Y]
		=\Tr\!\left((sI-\mathcal K_q(A;X))Y\right),
		\]
		proving \eqref{eq:first-variation}.
		
		If $X_*\in\Pm$ is a local minimizer of $\mathcal E_f$, then
		$D\mathcal E_f(X_*)[Y]=0$ for every $Y\in\Hm$.  Hence
		\[
		0
		=\Tr\!\left(
		\left[
		sI-\sum_{j=1}^n w_j\mathcal K_q(A_j;X_*)
		\right]Y
		\right)
		\]
		for every $Y\in\Hm$, which gives \eqref{eq:euler}.
		
		Conversely, if $\mathcal E_f$ is convex and $X_*$ satisfies
		\eqref{eq:euler}, then $D\mathcal E_f(X_*)=0$, and therefore
		\[
		\mathcal E_f(X)
		\ge\mathcal E_f(X_*)
		+D\mathcal E_f(X_*)[X-X_*]
		=\mathcal E_f(X_*)
		\]
		for every $X\in\Pm$.  Thus $X_*$ is a global minimizer.  If
		$\mathcal E_f$ is strictly convex, this minimizer, and hence the
		solution of \eqref{eq:euler}, is unique.
\end{proof}
	
For scalar data, the barycenter is unique and lies between the smallest and
largest data values, without any additional assumption on $f$.

\begin{proposition}
		\label{prop:scalar-barycenter}
		Let $m=1$, let $f$ be a nontrivial normalized operator monotone
		function, and let $a_1,\ldots,a_n>0$.  For positive weights summing to
		one, $\mathcal E_f$ has a unique minimizer $x_*\in(0,\infty)$.
		Moreover,
		\[
			\psi_f(r):=\frac{f(r)f'(r)}r
		\]
		is strictly decreasing, $\psi_f(1)=s$, and $x_*$ is the unique solution
		of
		\begin{equation}\label{eq:scalar-euler}
			\sum_{j=1}^n w_j
			\psi_f\!\left(\sqrt{\frac{x_*}{a_j}}\right)=s.
		\end{equation}
		It also satisfies
		\begin{equation}\label{eq:scalar-localization}
			\min_j a_j\le x_*\le\max_j a_j,
		\end{equation}
		with both inequalities strict unless all the data are equal.
	\end{proposition}
	
	\begin{proof}
		For $a,x>0$, put $r=\sqrt{x/a}$.  Then
		\[
			H_{f,a}'(x)=\psi_f(r),
			\qquad
			H_{f,a}''(x)
			=-\frac{f'(r)(f(r)-rf'(r))-rf(r)f''(r)}{2ar^3}.
		\]
		If $\mu_f$ is the representing measure and
		$d_\lambda(r)=(1-\lambda)r+\lambda$, then
		\begin{align*}
			f'(r)&=\int_{[0,1]}\frac{\lambda}{d_\lambda(r)^2}\,d\mu_f(\lambda),\\
			f(r)-rf'(r)
			&=\int_{[0,1]}\frac{(1-\lambda)r^2}{d_\lambda(r)^2}
			\,d\mu_f(\lambda).
		\end{align*}
		The two quantities are strictly positive because the excluded endpoint
		functions correspond to $\mu_f=\delta_0$ and $\mu_f=\delta_1$,
		respectively.  Since $f''\le0$, every $H_{f,a}$ is strictly concave;
		hence $\mathcal E_f$ is strictly convex.  Proposition~\ref{prop:closed-cone}
		supplies a nonzero closed-cone minimizer, which is therefore the unique
		minimizer $x_*\in(0,\infty)$.
		
		Because $H_{f,a}'(x)=\psi_f(\sqrt{x/a})$, strict concavity makes
		$\psi_f$ strictly decreasing; also $\psi_f(1)=s$.  Thus
		\[
			\mathcal E_f'(x)
			=s-\sum_{j=1}^n w_j
			\psi_f\!\left(\sqrt{\frac{x}{a_j}}\right),
		\]
		which proves \eqref{eq:scalar-euler}.  This derivative is negative below
		$\min_j a_j$ and positive above $\max_j a_j$, proving
		\eqref{eq:scalar-localization}.  At either endpoint the inequality is strict
		unless all $a_j$ coincide.
	\end{proof}
	
	\subsection{Residual certificates and the sharp power transition}
	
\begin{definition}\label{def:certificate}
		Let $m\ge1$ be an integer.  We say that $f$ satisfies the
		\emph{order-$m$ residual certificate} if there is a real constant
		$c\in\mathbb R$ such that
		\begin{equation}\label{eq:residual}
			G(x):=f(x)^2-cx^2
		\end{equation}
		is nonnegative, nonconstant, and matrix concave of order $m$ on
		$(0,\infty)$.  If, for the same constant $c$, the residual $G$ is a
		nonconstant operator monotone function on $(0,\infty)$, we say that $f$
		satisfies the \emph{dimension-free residual certificate}.  Here the word
		``certificate'' means that the stated property of the explicitly
		checkable residual $G$ will serve below as a sufficient condition for
		uniqueness and positive definiteness of the barycenter.
\end{definition}
	
We shall use the following standard consequence of positive matrix
concavity.

\begin{lemma}
	\label{lem:positive-matrix-concavity}
	Let $h:(0,\infty)\to[0,\infty)$ be matrix concave of order $m$.
	Then $h$ is matrix monotone of order $m$.  If, in addition, $h$
	is real analytic and nonconstant, then
	\[
	h'(x)>0,\qquad x>0.
	\]
\end{lemma}

\begin{proof}
	The first assertion is exactly
	\cite[Theorem~2.3]{Hansen2013}.  For the second, matrix monotonicity
	gives $h'\ge0$, while scalar concavity makes $h'$ nonincreasing.
	If $h'(x_0)=0$ at some $x_0>0$, then $h'$ vanishes on
	$(x_0,\infty)$.  Thus $h$ is constant on an open interval, and real
	analyticity forces it to be constant on $(0,\infty)$, a contradiction.
\end{proof}
	
	\begin{lemma}
		\label{lem:certificate-constant}
	 If $f$ satisfies the order-$m$ residual certificate, then the real
		constant $c$ in \eqref{eq:residual} is uniquely determined by $f$, and
		\begin{equation}\label{eq:canonical-c}
			c=\left(\lim_{x\to\infty}\frac{f(x)}{x}\right)^2
			=\mu_f(\{1\})^2,
			\qquad 0\le c<s<1,
		\end{equation}
		where $\mu_f$ is the representing measure in
		\eqref{eq:representing-measure}.  Moreover,
		\begin{equation}\label{eq:normalized-residual}
			\widetilde G(x):=\frac{G(x)}{1-c}
		\end{equation}
		is normalized, nonnegative, nonconstant, matrix concave of order $m$, and
		matrix monotone of order $m$.  Under the dimension-free certificate,
		$\widetilde G$ is operator monotone.
	\end{lemma}

 \begin{proof}
		Since $f(0+)$ is finite, we extend $f$ continuously and concavely to
		$[0,\infty)$ by setting $f(0)=f(0+)$.  If $0<x<y$, then concavity gives
		\[
		f(x)
		\ge \frac{x}{y}f(y)
		+\left(1-\frac{x}{y}\right)f(0),
		\]
		and hence
		\[
		\frac{f(x)-f(0)}{x}
		\ge
		\frac{f(y)-f(0)}{y}.
		\]
		Thus the function
		\[
		x\longmapsto\frac{f(x)-f(0+)}{x}
		\]
		is nonincreasing.  It is also nonnegative because $f$ is increasing.
		Consequently, the finite limit
		\[
		\ell:=\lim_{x\to\infty}\frac{f(x)}{x}\ge0
		\]
		exists.
		
		By the residual-certificate assumption, $G$ is nonnegative and
		scalar concave.  Since
		\[
		G(0+)=f(0+)^2<\infty,
		\]
		it extends continuously and concavely to $[0,\infty)$.  For $x>1$,
		write
		\[
		1=\left(1-\frac1x\right)0+\frac1x\,x.
		\]
		Concavity gives
		\[
		G(1)
		\ge
		\left(1-\frac1x\right)G(0+)+\frac1xG(x),
		\]
		or equivalently,
		\[
		G(x)
		\le
		G(1)+(x-1)\bigl(G(1)-G(0+)\bigr).
		\]
		The coefficient $G(1)-G(0+)$ must be nonnegative; otherwise, the
		right-hand side would be negative for all sufficiently large $x$,
		contradicting $G(x)\ge0$.  Therefore,
		\[
		0\le\frac{G(x)}{x^2}
		\le
		\frac{G(1)}{x^2}
		+\frac{x-1}{x^2}\bigl(G(1)-G(0+)\bigr)
		\longrightarrow0.
		\]
		Since
		\[
		\frac{G(x)}{x^2}
		=\left(\frac{f(x)}{x}\right)^2-c,
		\]
		letting $x\to\infty$ yields
		\[
		c=\ell^2.
		\]
		In particular, $c$ is uniquely determined by $f$ and satisfies
		$c\ge0$.
		
		We next identify $\ell$ in terms of the representing measure.  Dividing
		\eqref{eq:representing-measure} by $x$ gives
		\[
		\frac{f(x)}{x}
		=\int_{[0,1]}
		\frac{h_\lambda(x)}{x}\,d\mu_f(\lambda).
		\]
		For $x\ge1$,
		\[
		0\le\frac{h_\lambda(x)}{x}
		=\frac{1}{(1-\lambda)x+\lambda}\le1,
		\]
		and, for every $\lambda\in[0,1]$,
		\[
		\frac{h_\lambda(x)}{x}
		\longrightarrow
		\mathbf 1_{\{1\}}(\lambda)
		\qquad (x\to\infty).
		\]
		Since $\mu_f$ is a probability measure, the dominated convergence
		theorem gives
		\[
		\ell
		=\int_{[0,1]}\mathbf 1_{\{1\}}(\lambda)\,d\mu_f(\lambda)
		=\mu_f(\{1\}).
		\]
		Hence
		\[
		c=\ell^2=\mu_f(\{1\})^2.
		\]
		
		Since $f$ is operator monotone, it is real analytic on $(0,\infty)$,
		and therefore so is
		\[
		G(x)=f(x)^2-cx^2.
		\]
		Lemma~\ref{lem:positive-matrix-concavity} shows that $G$ is matrix
		monotone of order $m$ and, because it is nonconstant,
		\[
		G'(x)>0,\qquad x>0.
		\]
		In particular, using $f(1)=1$ and $f'(1)=s$, we obtain
		\[
		0<G'(1)=2f(1)f'(1)-2c=2(s-c).
		\]
		Thus
		\[
		0\le c<s<1.
		\]
		
		Finally,
		\[
		G(1)=f(1)^2-c=1-c>0.
		\]
		Hence
		\[
		\widetilde G(x):=\frac{G(x)}{1-c}
		\]
		is normalized, nonnegative, nonconstant, matrix concave of order $m$,
		and matrix monotone of order $m$.  Under the dimension-free residual
		certificate, $G$ is operator monotone, and therefore so is
		$\widetilde G$.
\end{proof}

For $m\times m$ data, the order-$m$ residual certificate guarantees a
unique positive definite global minimizer.

\begin{theorem}
	\label{thm:certificate}
	Let $m\ge1$, and suppose that $f$ satisfies the order-$m$ residual
	certificate of Definition~\ref{def:certificate}.  Then, for every
	integer $n\ge1$, every positive probability vector
	$\omega=(w_1,\ldots,w_n)$, and every
	$A_1,\ldots,A_n\in\Pm$, the continuous extension of
	$\mathcal E_f$ to $\overline{\mathbb{P}}_m$ has a unique global minimizer
	$X_*$.  Moreover, $X_*\in\Pm$ and is the unique solution of the
	barycenter equation \eqref{eq:euler}.
\end{theorem}
	
 \begin{proof}
		Fix $A\in\Pm$.  Since
		\[
		q(x)=f(x)^2=cx^2+G(x),
		\]
		functional calculus, cyclicity of the trace, and the Riccati identity
		$R_A(X)AR_A(X)=X$ give
		\begin{align}
			H_{f,A}(X)
			&=c\Tr\!\left(AR_A(X)^2\right)
			+\Tr\!\left(AG(R_A(X))\right) \notag\\
			&=c\Tr\!\left(R_A(X)AR_A(X)\right)
			+\Tr\!\left(AG(R_A(X))\right) \notag\\
			&=c\Tr X+\Tr\!\left(AG(R_A(X))\right).
			\label{eq:H-residual}
		\end{align}
		
		We first prove strict concavity.  Let $X,Y\in\Pm$ with $X\ne Y$,
		let $0<\theta<1$, and set
		\[
		X_\theta=(1-\theta)X+\theta Y,
		\qquad
		\overline R_\theta
		=(1-\theta)R_A(X)+\theta R_A(Y).
		\]
		Since $C_A$ is linear and the square-root map is operator concave,
		\begin{equation}\label{eq:strict-RA-concavity}
			R_A(X_\theta)\succeq\overline R_\theta.
		\end{equation}
		This inequality is strict in the sense that its two sides are not
		equal.  Indeed, put
		\[
		P=C_A(X)^{1/2},
		\qquad
		Q=C_A(Y)^{1/2}.
		\]
		Equality in \eqref{eq:strict-RA-concavity} would imply
		\[
		C_A(X_\theta)^{1/2}
		=(1-\theta)P+\theta Q.
		\]
		Squaring both sides and using
		\[
		C_A(X_\theta)
		=(1-\theta)P^2+\theta Q^2
		\]
		would give
		\[
		\theta(1-\theta)(P-Q)^2=0.
		\]
		Hence $P=Q$, so $C_A(X)=C_A(Y)$ and therefore $X=Y$, a
		contradiction.
		
		By Definition~\ref{def:certificate} and
		Lemma~\ref{lem:positive-matrix-concavity}, the residual $G$ is both
		matrix concave and matrix monotone of order $m$.  Therefore,
		\[
		G(R_A(X_\theta))
		\succeq G(\overline R_\theta)
		\succeq
		(1-\theta)G(R_A(X))+\theta G(R_A(Y)).
		\]
		Set
		\[
		D_1
		:=G(R_A(X_\theta))-G(\overline R_\theta)
		\]
		and
		\[
		D_2
		:=G(\overline R_\theta)
		-(1-\theta)G(R_A(X))-\theta G(R_A(Y)).
		\]
		Then $D_1,D_2\succeq0$.  Moreover, $D_1\ne0$.  Otherwise,
		\[
		G(R_A(X_\theta))=G(\overline R_\theta).
		\]
		By Lemma~\ref{lem:certificate-constant}, $G'(x)>0$ for every $x>0$,
		so $G$ is injective and has an inverse on its range.  Applying this
		inverse by functional calculus would yield
		\[
		R_A(X_\theta)=\overline R_\theta,
		\]
		contrary to the strictness of \eqref{eq:strict-RA-concavity}.
		
		It follows that $D:=D_1+D_2$ is positive semidefinite and nonzero.
		Since $A\succ0$,
		\[
		\Tr(AD)
		=\Tr\!\left(A^{1/2}DA^{1/2}\right)>0.
		\]
		Consequently,
		\[
		\Tr\!\left(AG(R_A(X_\theta))\right)
		>
		(1-\theta)\Tr\!\left(AG(R_A(X))\right)
		+\theta\Tr\!\left(AG(R_A(Y))\right).
		\]
		Thus
		\[
		X\longmapsto\Tr\!\left(AG(R_A(X))\right)
		\]
		is strictly concave.  Since $c\Tr X$ is affine,
		\eqref{eq:H-residual} shows that $H_{f,A}$ is strictly concave.
		Hence the positively weighted sum $H_{f,\omega}$ is strictly concave,
		and \eqref{eq:energy} shows that $\mathcal E_f$ is strictly convex
		on $\Pm$.
		
		We next show that no minimizer can lie on the boundary of the positive
		cone.  Let
		\[
		X_0\in\overline{\mathbb{P}}_m\setminus\Pm,
		\]
		so that $X_0\succeq0$ is noninvertible, and for $\varepsilon>0$ set
		\[
		X_\varepsilon=X_0+\varepsilon I,
		\qquad
		C_\varepsilon=C_A(X_\varepsilon),
		\qquad
		R_\varepsilon=R_A(X_\varepsilon).
		\]
		Fix $0<\varepsilon_0<1$.  The map
		\[
		\varepsilon\longmapsto R_A(X_0+\varepsilon I)
		\]
		is continuous on $[0,\varepsilon_0]$.  Its image is therefore
		bounded, so there exists $\kappa>0$ such that
		\[
		R_\varepsilon\preceq\kappa I,
		\qquad 0<\varepsilon\le\varepsilon_0.
		\]
		Put
		\[
		a:=\lambda_{\min}(A),
		\qquad
		b:=\lambda_{\max}(A),
		\qquad
		d_\kappa:=G'(\kappa)>0.
		\]
		
		Matrix monotonicity of order $m$ implies that
		$DG(R_\varepsilon)$ is a positive linear map on $\mathbb H_m$.
		Indeed, if $Z\succeq0$ and $t>0$, then
		\[
		R_\varepsilon\preceq R_\varepsilon+tZ,
		\]
		so
		\[
		\frac{G(R_\varepsilon+tZ)-G(R_\varepsilon)}{t}\succeq0.
		\]
		Letting $t\downarrow0$ and using the closedness of the positive
		semidefinite cone gives
		\[
		DG(R_\varepsilon)[Z]\succeq0.
		\]
		In particular, since $A\succeq aI$,
		\[
		DG(R_\varepsilon)[A]
		\succeq aDG(R_\varepsilon)[I]
		=aG'(R_\varepsilon).
		\]
		Scalar concavity makes $G'$ nonincreasing.  Since
		$\sigma(R_\varepsilon)\subset(0,\kappa]$, functional calculus gives
		\[
		G'(R_\varepsilon)\succeq G'(\kappa)I=d_\kappa I.
		\]
		Using also $A^{-1}\succeq b^{-1}I$, we obtain
		\begin{align}
			A^{-1/2}DG(R_\varepsilon)[A]A^{-1/2}
			&\succeq
			aA^{-1/2}G'(R_\varepsilon)A^{-1/2} \notag\\
			&\succeq
			ad_\kappa A^{-1} \notag\\
			&\succeq
			\frac{a}{b}d_\kappa I.
			\label{eq:DG-lower-bound}
		\end{align}
		
		By \eqref{eq:Kg}, positivity of
		$\mathcal S_{C_\varepsilon}$ and
		\[
		\mathcal S_{C_\varepsilon}[I]
		=\frac12C_\varepsilon^{-1/2}
		\]
		now yield
		\[
		\mathcal K_G(A;X_\varepsilon)
		\succeq
		\frac{ad_\kappa}{2b}
		A^{1/2}C_\varepsilon^{-1/2}A^{1/2}.
		\]
		Taking traces and using $A\succeq aI$, we obtain
		\begin{align}
			\Tr\mathcal K_G(A;X_\varepsilon)
			&\ge
			\frac{ad_\kappa}{2b}
			\Tr\!\left(AC_\varepsilon^{-1/2}\right) \notag\\
			&\ge
			\frac{a^2d_\kappa}{2b}
			\Tr C_\varepsilon^{-1/2}.
			\label{eq:KG-boundary-blowup}
		\end{align}
		
		Since $A\succ0$, the congruence
		\[
		C_A(X_0)=A^{1/2}X_0A^{1/2}
		\]
		is singular whenever $X_0$ is singular.  Moreover,
		\[
		C_\varepsilon=C_A(X_0)+\varepsilon A
		\longrightarrow C_A(X_0).
		\]
		Hence
		\[
		\lambda_{\min}(C_\varepsilon)\longrightarrow0,
		\]
		and therefore
		\[
		\Tr C_\varepsilon^{-1/2}
		\ge
		\lambda_{\min}(C_\varepsilon)^{-1/2}
		\longrightarrow+\infty.
		\]
		It follows from \eqref{eq:KG-boundary-blowup} that
		\[
		\Tr\mathcal K_G(A;X_\varepsilon)
		\longrightarrow+\infty.
		\]
		
		Differentiating \eqref{eq:H-residual} in the direction $I$ gives
		\[
		DH_{f,A}(X_\varepsilon)[I]
		=mc+\Tr\mathcal K_G(A;X_\varepsilon)
		\longrightarrow+\infty.
		\]
		Applying this argument to each $A_j$, summing with the positive
		weights, and using \eqref{eq:energy}, we obtain
		\[
		\frac{d}{d\varepsilon}
		\mathcal E_f(X_0+\varepsilon I)
		=
		sm-DH_{f,\omega}(X_0+\varepsilon I)[I]
		\longrightarrow-\infty.
		\]
		Thus there exists $\delta>0$ such that
		\[
		\frac{d}{d\varepsilon}
		\mathcal E_f(X_0+\varepsilon I)\le-1,
		\qquad 0<\varepsilon<\delta.
		\]
		Define
		\[
		\varphi(t):=\mathcal E_f(X_0+tI).
		\]
		For $0<\eta<\varepsilon<\delta$,
		\[
		\varphi(\varepsilon)-\varphi(\eta)
		=\int_\eta^\varepsilon\varphi'(t)\,dt
		\le-(\varepsilon-\eta).
		\]
		Letting $\eta\downarrow0$ and using continuity of the closed-cone
		extension gives
		\[
		\mathcal E_f(X_0+\varepsilon I)
		-\mathcal E_f(X_0)
		\le-\varepsilon<0.
		\]
		Hence no matrix in
		$\overline{\mathbb{P}}_m\setminus\Pm$ can minimize $\mathcal E_f$.
		
		Proposition~\ref{prop:closed-cone} supplies a global minimizer in
		$\overline{\mathbb{P}}_m$, and the preceding argument shows that it belongs to
		$\Pm$.  Strict convexity makes this minimizer unique.
		Proposition~\ref{prop:euler} then shows that it satisfies
		\eqref{eq:euler}; conversely, every solution of \eqref{eq:euler} is
		this unique minimizer.
\end{proof}
	
	\begin{corollary}
		\label{cor:dimension-free-certificate}
		Suppose that there is $c\in\mathbb R$ such that
		\[
			G(x)=f(x)^2-cx^2
		\]
		is a nonconstant operator monotone function on $(0,\infty)$.  Then
		Theorem~\ref{thm:certificate} holds in every matrix size.
		
		Conversely, if the order-$m$ certificate holds for every $m$ with
		the same constant $c$, then $G$ is operator monotone.  In this
		dimension-free setting, the certificate holds with $c=0$ exactly for
		\[
			f=\sqrt g,
		\]
		where $g$ is a normalized nonconstant operator monotone function.
	\end{corollary}
	
	\begin{proof}
		Since $G(0+)=f(0+)^2\ge0$ and $G$ is increasing, $G$ is
		nonnegative.  Every operator monotone function on the positive half-line
		is operator concave, so Definition~\ref{def:certificate} applies in every
		finite size.  Conversely, certificates of every order make $G$ matrix
		monotone of every order by
		Lemma~\ref{lem:positive-matrix-concavity}, hence operator monotone.
		
		If $c=0$, take $g=f^2=G$.  Conversely, if $g$ is normalized,
		nonconstant, and operator monotone, then $f=\sqrt g$ is normalized and
		operator monotone because both $g$ and the square-root function preserve
		operator order; moreover, $f^2=g$.
	\end{proof}
	
	\begin{remark}
		\label{rem:certificate-examples}
	Several concrete classes connect the criterion directly with the preceding
	results.  For $f(x)=x^t$ with $0<t\le1/2$, one has
	$f(x)^2=x^{2t}$, which is a nonconstant operator monotone function.
	Thus the dimension-free certificate holds with $c=0$.
	
	More generally, let
	\[
	f(x)=g_{r,t}(x):=(1-r)+rx^t,
	\qquad
	0<r<1,\quad 0<t\le\frac12.
	\]
	Then $g_{r,t}$ is a nontrivial normalized operator monotone function, and
	\[
	g_{r,t}(x)^2
	=(1-r)^2+2r(1-r)x^t+r^2x^{2t}.
	\]
	Since both $x^t$ and $x^{2t}$ are operator monotone in the stated
	parameter range, $g_{r,t}^2$ is a nonconstant operator monotone function.
	Hence the dimension-free certificate again holds with $c=0$.
	
Every nontrivial exact-DPI function $f=f_{u,v}$ identified in
Theorem~\ref{thm:quadratic-square-classification} also satisfies the
dimension-free certificate, since
\[
f_{u,v}(x)^2-uvx^2
=(1-u)(1-v)+(u+v-2uv)x
\]
is affine with strictly positive slope.  In particular, for
	$f(x)=(1-s)+sx$, the certificate holds with $c=s^2$, and
	\eqref{eq:euler} reduces to the standard Bures--Wasserstein barycenter
	equation.
		
		The residual certificate is not needed for scalar data:
		Proposition~\ref{prop:scalar-barycenter} applies to every nontrivial
		normalized operator monotone function.  The next theorem shows that the
		power-family range is sharp.  For general $f$, no necessity of the
		finite-order certificate for arbitrary-data uniqueness is asserted.
	\end{remark}
	
 \begin{theorem}
		\label{thm:power-barycenter-phase-transition}
		Let $f_t(x)=x^t$, where $0<t<1$.  The following are equivalent.
		\begin{enumerate}
			\item[\rm(i)] $0<t\le\frac12$.
			
			\item[\rm(ii)] For all integers $m,n\ge1$, every positive probability
			vector $\omega=(w_1,\ldots,w_n)$, and every
			$A_1,\ldots,A_n\in\mathbb P_m$, the continuous extension of
			$\mathcal E_{f_t}$ to $\overline{\mathbb P_m}$ has a unique global
			minimizer, and this minimizer belongs to $\mathbb P_m$.
		\end{enumerate}
		More precisely, if $\frac12<t<1$, then, for all sufficiently large
		$L>1$, the equal-weight pair
		\[
		A_1=\begin{pmatrix}L&0\\0&1\end{pmatrix},
		\qquad
		A_2=\begin{pmatrix}1&0\\0&L\end{pmatrix}
		\]
		has at least two distinct global minimizers on
		$\overline{\mathbb P_2}$.
	\end{theorem}
	
	\begin{proof}
		Suppose first that $0<t\le\frac12$.  Then
		\[
		f_t(x)^2=x^{2t}
		\]
		is a nonconstant operator monotone function.  Thus $f_t$ satisfies
		the dimension-free residual certificate with $c=0$, and
		Corollary~\ref{cor:dimension-free-certificate} proves ${\rm(ii)}$.
		
		Conversely, fix $\frac12<t<1$.  For $L>1$, let
		\[
		A_1=\begin{pmatrix}L&0\\0&1\end{pmatrix},
		\qquad
		A_2=\begin{pmatrix}1&0\\0&L\end{pmatrix},
		\qquad
		w_1=w_2=\frac12,
		\]
		and define
		\[
		H_{t,L}(X)
		:=\frac12\sum_{j=1}^2
		\Tr\!\left(A_jR_{A_j}(X)^{2t}\right).
		\]
		The corresponding barycenter functional is
		\begin{equation}\label{eq:power-counterexample-energy}
			\mathcal E_{t,L}(X)
			=(1-t)(L+1)+t\Tr X-H_{t,L}(X).
		\end{equation}
		
		We first minimize the functional over the diagonal cone.  If
		\[
		A=\operatorname{diag}(a_1,a_2),
		\qquad
		X=\operatorname{diag}(x_1,x_2)\succeq0,
		\]
		then
		\[
		R_A(X)
		=\operatorname{diag}\left(
		\sqrt{\frac{x_1}{a_1}},
		\sqrt{\frac{x_2}{a_2}}
		\right).
		\]
		Applying this formula to $A_1$ and $A_2$ gives
		\[
		H_{t,L}(X)
		=b_L(x_1^t+x_2^t),
		\qquad
		b_L:=\frac{L^{1-t}+1}{2}.
		\]
		Consequently,
		\[
		\mathcal E_{t,L}(X)
		=(1-t)(L+1)
		+\sum_{k=1}^2\bigl(tx_k-b_Lx_k^t\bigr).
		\]
		
		For
		\[
		\phi_L(x):=tx-b_Lx^t,
		\]
		we have
		\[
		\phi_L'(x)
		=t\bigl(1-b_Lx^{t-1}\bigr),
		\qquad
		\phi_L''(x)
		=b_Lt(1-t)x^{t-2}>0.
		\]
		Moreover,
		\[
		\phi_L'(x)\longrightarrow-\infty
		\quad\text{as }x\downarrow0,
		\qquad
		\phi_L'(x)\longrightarrow t
		\quad\text{as }x\to\infty.
		\]
		Hence $\phi_L$ has the unique minimizer
		\[
		\rho_L=b_L^{1/(1-t)}.
		\]
		Therefore, the restriction of $\mathcal E_{t,L}$ to the closed
		diagonal cone has the unique minimizer
		\[
		X_L^{\rm diag}=\rho_LI.
		\]
		
		We now compare the diagonal minimizer with a rank-one direction.  Put
		\[
		v=\frac1{\sqrt2}\binom11,
		\qquad
		P=vv^*,
		\qquad
		Y_0=\frac12I.
		\]
		For every $A\succ0$, write $w=A^{1/2}v$.  Then
		\[
		C_A(P)=A^{1/2}PA^{1/2}=ww^*.
		\]
		Since
		\[
		(ww^*)^{1/2}=\frac{ww^*}{\|w\|}
		\qquad\text{and}\qquad
		\|w\|^2=v^*Av,
		\]
		we obtain
		\begin{equation}\label{eq:rank-one-riccati-root}
			R_A(P)
			=A^{-1/2}C_A(P)^{1/2}A^{-1/2}
			=\frac{P}{\sqrt{v^*Av}}.
		\end{equation}
		Because $P^2=P$,
		\[
		R_A(P)^{2t}=(v^*Av)^{-t}P,
		\]
		and therefore
		\[
		\Tr\!\left(AR_A(P)^{2t}\right)
		=(v^*Av)^{1-t}.
		\]
		It follows that
		\[
		H_{t,L}(P)
		=\left(\frac{L+1}{2}\right)^{1-t}.
		\]
		A direct diagonal calculation also gives
		\[
		H_{t,L}(Y_0)
		=2^{-t}\bigl(L^{1-t}+1\bigr).
		\]
		
		For every $\lambda\ge0$ and $Y\succeq0$, homogeneity of $R_A$ gives
		\[
		R_A(\lambda Y)=\sqrt{\lambda}\,R_A(Y),
		\]
		and hence
		\[
		H_{t,L}(\lambda Y)=\lambda^tH_{t,L}(Y).
		\]
		In particular, if $\Tr Y=1$, then
		\[
		\mathcal E_{t,L}(\lambda Y)
		=(1-t)(L+1)+t\lambda
		-\lambda^tH_{t,L}(Y).
		\]
		The function of $\lambda$ on the right is strictly convex on
		$(0,\infty)$, and its derivative vanishes exactly when
		\[
		\lambda^{1-t}=H_{t,L}(Y).
		\]
		Thus its unique minimizer is
		\[
		\lambda_Y=H_{t,L}(Y)^{1/(1-t)},
		\]
		and
		\begin{equation}\label{eq:power-ray-minimum}
			\min_{\lambda\ge0}\mathcal E_{t,L}(\lambda Y)
			=(1-t)(L+1)
			-(1-t)H_{t,L}(Y)^{1/(1-t)}.
		\end{equation}
		
		Finally,
		\[
		\frac{H_{t,L}(P)}{H_{t,L}(Y_0)}
		=
		2^{2t-1}
		\frac{(L+1)^{1-t}}{L^{1-t}+1}
		\longrightarrow 2^{2t-1}>1
		\qquad (L\to\infty).
		\]
		Hence, for all sufficiently large $L$,
		\[
		H_{t,L}(P)>H_{t,L}(Y_0).
		\]
		By \eqref{eq:power-ray-minimum}, the minimum along a ray is strictly
		decreasing as $H_{t,L}(Y)$ increases.  Therefore,
		\[
		\min_{\lambda\ge0}\mathcal E_{t,L}(\lambda P)
		<
		\min_{\lambda\ge0}\mathcal E_{t,L}(\lambda Y_0).
		\]
		The ray $\{\lambda Y_0:\lambda\ge0\}$ is precisely the scalar ray,
		since $\lambda Y_0=(\lambda/2)I$, and it contains
		$X_L^{\rm diag}=\rho_LI$.  Its minimum therefore equals the minimum
		of $\mathcal E_{t,L}$ over the entire diagonal cone.  It follows that
		no diagonal matrix can be a global minimizer.
		
		By Proposition~\ref{prop:closed-cone}, a global minimizer
		$X_*\succeq0$ nevertheless exists.  Let
		\[
		J=\operatorname{diag}(1,-1).
		\]
		Since $J$ commutes with $A_1$ and $A_2$, we have
		\[
		C_{A_j}(JXJ)=JC_{A_j}(X)J
		\]
		and hence, by functional calculus,
		\[
		R_{A_j}(JXJ)=JR_{A_j}(X)J,
		\qquad j=1,2.
		\]
		Unitary invariance of the trace now gives
		\[
		\mathcal E_{t,L}(JXJ)=\mathcal E_{t,L}(X).
		\]
		Thus $JX_*J$ is also a global minimizer.
		
		A Hermitian $2\times2$ matrix satisfies $JXJ=X$ if and only if it is
		diagonal.  Since no global minimizer is diagonal, we have
		\[
		JX_*J\ne X_*.
		\]
		Therefore $X_*$ and $JX_*J$ are two distinct global minimizers, so
		${\rm(ii)}$ fails whenever $\frac12<t<1$.
\end{proof}
	
The preceding argument establishes nonuniqueness for the continuous
closed-cone problem, but it does not show that the minimizers are positive
definite.  For the open-cone problem, therefore, either the infimum is not
attained in $\mathbb P_2$, or it is attained by at least two distinct
positive definite matrices related by the symmetry $X\mapsto JXJ$.
Thus these data cannot have a unique positive definite barycenter.
	
	\section{Concluding remarks}
	
	The sharp comparison and Hessian formula show that every alternative-mean
	divergence considered here is globally comparable with the Bures--Wasserstein distance, with diagonal
	Hessian governed by the corresponding metric. Exact data processing is much more rigid: the positive trace-preserving
	and the usual CPTP formulations coincide, are equivalent to joint
	convexity, and are completely detected by diagonal pinching on faithful
	qubit states.  The corresponding representing functions form a
	two-parameter family, and their divergences are precisely scalar multiples
	of squared Bures--Wasserstein distance.
	
For barycenters, a global minimizer on the closed cone always exists, and
scalar data have a unique positive minimizer, without additional assumptions
on $f$.  A nonnegative, nonconstant, order-$m$ matrix-concave residual
gives a unique positive definite minimizer for $m\times m$ data, while an
operator-monotone residual gives the same conclusion in every matrix size.
For the power family, the property that every data set has a unique
closed-cone minimizer belonging to $\Pm$ holds exactly for
$0<t\le1/2$.  When $1/2<t<1$, nonuniqueness of the closed-cone problem
already occurs for commuting qubit data.  It remains open to characterize all normalized operator monotone functions
$f$ for which $\mathcal E_f$ has a unique positive definite minimizer
in every matrix size, for every finite positive definite data set and every
positive probability vector.

\section*{Acknowledgments}
The third author was partially supported by the KAKENHI Grant Number JP25K07036.

\end{document}